%% file: schuetzenberger.tex
\documentclass[runningheads,a4paper,%
  orivec,envcountsame,envcountsect]{llncs}
\usepackage[utf8]{inputenc}
\usepackage[T1]{fontenc}
\usepackage[british]{babel}
\usepackage{xspace}
\usepackage{dsfont}
\usepackage{tabularx}

\usepackage{paralist}

\usepackage[%
rm={oldstyle=false,proportional=true},%
sf={oldstyle=false,proportional=true},%
tt={oldstyle=false,proportional=true,variable=true},%
qt=false%
]{cfr-lm}

\usepackage[noadjust]{cite} 
\usepackage{graphicx,paralist,cite,csquotes}
\usepackage[pdftex,colorlinks=false]{hyperref}

\usepackage{amssymb,mathrsfs,mathtools,bm}
\numberwithin{equation}{section}
\usepackage[all,2cell]{xy}
\UseAllTwocells
\SelectTips{cm}{}
\newdir{ >}{{}*!/-5pt/@{>}}
\newdir{ (}{{}*!/-10pt/\dir^{(}}
\newdir{ )}{{}*!/-10pt/\dir_{(}}

\let\doendproof\endproof
\renewcommand\endproof{~\hfill\qed\doendproof}

\usepackage[draft]{fixme} 
\FXRegisterAuthor{hu}{ahu}{HU} 
\FXRegisterAuthor{ltc}{altc}{LTC} 

\usepackage{microtype}
\allowdisplaybreaks

\spnewtheorem{assumptions}[theorem]{Assumptions}{\bfseries}{\rmfamily}
\spnewtheorem{notation}[theorem]{Notation}{\bfseries}{\rmfamily}
\spnewtheorem{observation}[theorem]{Observation}{\bfseries}{\rmfamily}
\spnewtheorem{defn}[theorem]{Definition}{\bfseries}{\rmfamily}
\spnewtheorem{expl}[theorem]{Example}{\bfseries}{\rmfamily}
\spnewtheorem{rem}[theorem]{Remark}{\bfseries}{\rmfamily}
\spnewtheorem{construction}[theorem]{Construction}{\bfseries}{\rmfamily}
\spnewtheorem{appsec}[theorem]{}{\bfseries}{\rmfamily}

\renewcommand{\S}{\mathbb{S}}
\newcommand{\K}{\mathds{K}}

\newcommand{\dash}{\mathord{-}}

\newcommand{\Set}{\mathbf{Set}}
\newcommand{\Pos}{\mathbf{Pos}}

\newcommand{\JSL}{\mathbf{JSL}}
\newcommand{\Vect}[1]{#1\text{-}\mathbf{Vec}}

\newcommand{\Alg}[1]{#1\text{-}\mathbf{Alg}}

\newcommand{\A}{\mathscr{A}}
\newcommand{\B}{\mathscr{B}}
\newcommand{\C}{\mathscr{C}}
\newcommand{\D}{\mathscr{D}}
\newcommand{\X}{\mathscr{X}}

\newcommand{\V}{\mathcal{V}}

\newcommand{\ol}{\overline}

\newcommand{\dnarrow}{\mathord{\downarrow}}

\renewcommand{\epsilon}{\varepsilon}
\newcommand{\id}{\mathit{id}}
\newcommand{\Id}{\mathsf{Id}}
\newcommand{\seq}{\subseteq}
\newcommand{\xra}{\xrightarrow}

\newcommand{\defeq}{\coloneqq}

\renewcommand{\o}{\cdot}
\renewcommand{\t}{\otimes}
\newcommand{\takeout}[1]{\empty}

\renewcommand{\phi}{\varphi}
\newcommand{\ra}{\rightarrow}

\newcommand{\Ra}{\Rightarrow}

\newcommand{\Lra}{\Leftrightarrow}

\newcommand{\BA}{\mathbf{BA}}

\newcommand{\DL}{\mathbf{DL}}

\newcommand{\Rec}[1]{\mathbf{Rec}(#1)}
\newcommand{\Dn}{\mathcal{D}}

\newcommand{\Mon}[1]{\mathbf{Mon}(#1)}

\newcommand{\SMod}[1]{\mathbf{Mod}(#1)}
\newcommand{\Mod}[1]{#1\text{-}\mathbf{Mod}}

\newcommand{\under}[1]{|#1|}

\newcommand{\epito}{\twoheadrightarrow}

\newcommand{\monoto}{\rightarrowtail}

\newcommand{\Pow}{\mathcal{P}}

\newcommand{\one}{\mathbf{1}}

\renewcommand{\L}{\mathbb{L}}
\renewcommand{\c}{\circ}

\title{Sch\"utzenberger Products in a Category}
\titlerunning{Sch\"utzenberger Products in a Category}

\author{Liang-Ting Chen\inst{1}\fnmsep\thanks{Liang-Ting Chen acknowledges support from AFOSR.} \and
  Henning Urbat\inst{2}\fnmsep\thanks{Henning Urbat acknowledges support from 
  DFG under project AD 187/2-1.}}
\authorrunning{L.-T.~Chen, H.~Urbat}

\institute{Department of Information and Computer Sciences \\ University of
  Hawaii at Manoa, Honolulu, HI, USA
  \and Institut f\"{u}r Theoretische Informatik \\ Technische Universit\"at
  Braunschweig, Germany}

\begin{document}
\maketitle

\input{maintext.tex}

\bibliographystyle{splncs03}
\bibliography{ourpapers,coalgebra}

\newpage
\appendix
\input{appendix.tex}

\end{document}

%% file: maintext.tex
\begin{abstract}
The Sch\"utzenberger product of monoids is a key tool for the algebraic 
treatment of 
language concatenation. In this paper we generalize the Sch\"utzenberger product to the level of monoids in
an algebraic category $\D$, leading to a uniform view of the corresponding 
constructions for 
monoids (Sch\"utzenberger), ordered monoids (Pin), idempotent semirings 
(Kl\'ima and Pol\'ak) and algebras over a field (Reutenauer).
In addition, assuming that $\D$ is part of a Stone-type duality, we derive a 
characterization 
of the languages recognized by Sch\"utzenberger products.
\end{abstract}

\section{Introduction}
Since the early days of automata theory, it has been known that regular 
languages are precisely the languages 
recognized by finite monoids. This observation is the origin of 
algebraic language theory. One of the classical and ongoing challenges of this 
theory is 
the algebraic treatment of the concatenation of languages. 
The most important
tool for this purpose is the \emph{Sch\"utzenberger product} $M\diamond N$ 
of   
two monoids $M$ and $N$, introduced in \cite{sch65}. Its key 
property is that it recognizes all marked products of languages 
recognized by $M$ and $N$. Later, 
Reutenauer \cite{reutenauer79} showed 
that $M\diamond N$ is the ``smallest''  monoid with this property: any 
language recognized by $M\diamond N$ is a 
boolean combination of such marked products.

In the past decades, the original notion of language recognition by finite 
monoids has 
been refined to other algebraic structures, namely to ordered monoids 
by Pin~\cite{pin95}, to idempotent semirings by Pol\'ak~\cite{polak01}, and to 
associative algebras over a 
field 
by Reutenauer~\cite{reu80}. For all these structures, 
a Sch\"utzenberger 
product was introduced separately~\cite{pin03,klimapolak10,reu80}. 
Moreover,
Reutenauer's characterization of the languages recognized by Sch\"utzenberger 
products has been adapted to ordered monoids and idempotent semirings, 
replacing 
boolean combinations by positive boolean 
combinations~\cite{pin03} and finite unions~\cite{klimapolak10}, respectively.

This paper presents a unifying approach to Sch\"utzenberger products, 
covering the aforementioned constructions and results as special cases. Our 
starting point is the observation that all the algebraic structures appearing 
above (monoids, ordered monoids, idempotent semirings, and algebras over a 
field $\K$)  are 
\emph{monoids} interpreted in some variety $\D$ of algebras or 
ordered 
algebras, viz.\ $\D=$ sets, posets, semilattices, and $\K$-vector spaces, 
respectively. 
Next, we note that these categories $\D$ are related to the
category $\Mod{\S}$ of modules over some semiring $\S$. Indeed, semilattices and 
vector spaces are precisely modules over the two-element idempotent 
semiring
$\S=\{0,1\}$ and the field $\S=\K$, respectively. And every set 
or poset freely
\emph{generates} a semilattice (i.e.\ a module over $\{0,1\}$), viz. 
the semilattice
of finite subsets or finitely generated down-sets. Precisely speaking,
each of the above categories $\D$ admits a \emph{monoidal adjunction}
 \begin{equation}\label{eq:monoidaladj}
 \xymatrix@+1em{
   \Mod{\S} \ar@<1.2ex>[r]^-{U}
 \ar@{}[r]|-{\top} &
 \D 
 \ar@<1.2ex>[l]^-{F}}
 \end{equation}
for some semiring $\S$, where $U$ is a 
 forgetful functor and $F$ is a free construction. 

In this paper we introduce the Sch\"utzenberger 
product at the level of an 
abstract monoidal adjunction  \eqref{eq:monoidaladj}: for any two 
$\D$-monoids $M$ and 
$N$, we construct a $\D$-monoid $M\diamond N$ that recognizes all 
marked products of languages recognized by $M$ and $N$ (Theorem 
\ref{thm:schurec}), and prove that $M\diamond N$ 
is the ``smallest'' 
$\D$-monoid with this property (Theorem \ref{thm:universalprop}). 
Further, we derive a characterization of the languages recognized 
by $M\diamond N$ in
the spirit of Reutenauer's theorem \cite{reutenauer79}. To this end, we 
consider another variety $\C$ that is 
\emph{dual} to $\D$ on the level 
of finite 
algebras. For example, for $\D=$ sets we choose $\C=$ boolean algebras, since Stone's representation theorem gives a dual
equivalence between finite boolean algebras and finite sets.
We then prove that every language recognized by 
$M\diamond N$ is a 
``$\C$-algebraic combination'' of languages recognized by $M$ 
and $N$ and 
their marked products (Theorem \ref{thm:schurec2}). The explicit 
use of duality 
makes our proof conceptually different 
from the original ones.

By instantiating \eqref{eq:monoidaladj} to the proper adjunctions, we 
recover the Sch\"utzenberger product for monoids, ordered monoids, idempotent 
semirings and algebras over a field, and obtain a new Sch\"utzenberger 
product 
for algebras over a commutative semiring. Moreover, our Theorems 
\ref{thm:schurec} and \ref{thm:schurec2} specialize to the corresponding 
results
\cite{reutenauer79,pin03,klimapolak10} for (ordered) monoids and idempotent 
semirings. In the 
case of $\K$-algebras,
Theorem \ref{thm:schurec2} appears to be a new result. Apart from that, we 
believe that the main contribution of our paper is the identification of a categorical setting for language concatenation. We hope that the 
generality and the conceptual nature of our approach can contribute to an improved understanding
of the various ad hoc constructions and separate results appearing in the
literature.

\paragraph{Related work.} In recent years, categorical approaches to 
algebraic language theory have been a growing research topic. The present 
paper is a natural continuation of \cite{amu15}, where
we showed that the construction of syntactic monoids works at the level 
of  
$\D$-monoids in any commutative variety $\D$, allowing for a uniform treatment of 
syntactic (ordered) monoids, idempotent semirings and algebras over a 
field. The systematic use of duality in algebraic language theory 
originates in the work of Gehrke, 
Grigorieff, and Pin \cite{ggp08}, who interpreted Eilenberg's variety theorem 
in terms of Stone duality. In our papers
\cite{ammu14,ammu15,cu15} we extended their approach to an abstract 
Stone-type 
duality, leading to a uniform view of several Eilenberg-type theorems for 
regular 
languages. See also \cite{bcr15,uramoto16} for related 
duality-based work. Recently, Boja\'nczyk \cite{boj15} proposed to use 
\emph{monads} 
instead of monoids to get a categorical grasp on languages beyond finite 
words. By 
combining this idea with our duality framework, we established in 
\cite{camu16,uacm16} a variety theorem that covers 
most Eilenberg-type correspondences known in the literature, e.g. for
 languages 
of finite words, infinite words, words on linear orderings, trees, and cost 
functions.

\section{Preliminaries}
In this paper we study monoids and language recognition in algebraic categories.
The reader is assumed to be familiar with basic universal algebra and category
theory; see the Appendix for a toolkit. We call a
variety $\D$ of algebras or ordered algebras  \emph{commutative} if, for any
two algebras $A,B\in \D$, the set $[A,B]$ of morphisms from $A$ to $B$ forms an
algebra of $\D$ with operations taken pointwise in $B$. Our applications involve
the commutative varieties $\Set$  (sets), $\Pos$ (posets, as
ordered algebras without any operation), $\JSL$ (join-semilattices with $0$),
$\Vect{\K}$ (vector spaces over a field $\K$) and
$\Mod{\S}$ (modules over a commutative semiring $\S$ with $0, 1$). Note that $\JSL$ and 
$\Vect{\K}$ are special cases of $\Mod{\S}$ for $\S=\{0,1\}$, the two-element
semiring with $1+1=1$, and $\S=\K$, respectively.

\begin{notation}
Let $\A$, $\B$, $\C$, $\D$ always denote 
commutative varieties of algebras or ordered algebras. We write 
$\Psi=\Psi_\D\colon 
\Set\ra\D$ for the left adjoint to the forgetful functor 
$\under{\mathord{-}}\colon\D\ra\Set$; thus $\Psi X$ is the free algebra of $\D$ 
over~$X$. For simplicity, we assume that $X$ is a subset 
of $\under{\Psi X}$ and the universal map $X\monoto \under{\Psi X}$ is the inclusion. Denote by $\one_\D = \Psi 1$ the free 
one-generated algebra.
\end{notation}

\begin{expl}\label{ex:psi}
\begin{enumerate}[(1)]
\item For $\D=\Set$ or $\Pos$ we have $\Psi X = X$ (discretely ordered). 
\item For $\D = \JSL$ we get $\Psi X 
= (\Pow_f X,\cup)$, the semilattice of finite 
subsets of $X$.
\item For $\D=\Mod{\S}$ we have $\Psi X=
\S^{(X)}$, the $\S$-module of all finite-support
functions $X\ra \S$ with sum and scalar product defined pointwise. 
\end{enumerate}
\end{expl}

\begin{defn} Let $A,B,C\in \D$. By a 
\emph{bimorphism} from 
$A$, $B$ to 
$C$ is meant a function $f\colon \under{A}\times\under{B}\ra\under{C}$ such that the maps 
$f(a,\mathord{-})\colon \under{B}\ra \under{C}$ and $f(\mathord{-},b)\colon \under{A}\ra 
\under{C}$ carry morphisms of $\D$ for every $a\in\under{A}$ and $b\in 
\under{B}$. A \emph{tensor product} of $A$ and $B$ is a universal bimorphism 
$t_{A,B}\colon \under{A}\times \under{B} \ra \under{A\t B}$, in the sense that
for any bimorphism $f\colon \under{A}\times \under{B} \ra \under{C}$
there is a 
unique  $f'\colon A\t B \ra C$ in $\D$ with $f'\c t_{A,B} = f$. We denote by $a\t 
b$ the element $t_{A,B}(a,b)\in\under{A\t B}$.
\end{defn} 

\begin{expl}
In $\Set$ and $\Pos$ we have $A\t B= A\times B$.
In $\Mod{\S}$, $A\t B$ is the usual tensor product of $\S$-modules, 
and $t_{A,B}$ is the universal $\S$-bilinear map.
\end{expl}

\begin{rem}\label{rem:tensorproducts}
\begin{enumerate}[(1)]
\item Tensor products exist in any commutative variety $\D$, see \cite{bn76}.
\item $\t$ is associative and has unit 
$\one_\D$; that is, there are natural isomorphisms  
\[\alpha_{A,B,C}\colon (A\t B)\t C \cong A\t (B\t C),\quad \rho_A\colon A\t \one_\D 
\cong A,\quad \lambda_A\colon \one_\D\t A \cong A.\]
\item Given $f\colon A\to C$ and $g\colon B\to D$ in $\D$,  denote by $f\t g\colon A\t B\to C\t 
D$ the morphism induced by the bimorphism 
$\under{A}\times\under{B}\xra{f\times g} \under{C}\times\under{D} 
\xra{t_{C,D}}\under{C\t D}$.
\end{enumerate}
\end{rem}

\begin{defn}
  \label{def:dmonoid}
A \emph{$\D$-monoid} is a triple 
$(M, 1, \bullet)$ where $M$ is an object of $\D$ and $(\under{M}, 1,\bullet)$ 
is a monoid whose multiplication
$\under{M}\times\under{M}\xra{\bullet}\under{M}$ is a bimorphism of $\D$. A
\emph{morphism} $h\colon (M, 1_M,\bullet_M)\ra (N, 1_N,\bullet_N)$ of
$\D$-monoids is a morphism $h\colon M\ra N$ in $\D$ with $h(1_M)=1_N$ and 
$h(m\bullet_M m') = h(m)\bullet_N h(m')$ for $m,m'\in\under{M}$.
We denote
the category of $\D$-monoids by $\Mon{\D}$.
\end{defn}

\begin{expl}\label{ex:dmon}
Monoids in $\D=\Set$, $\Pos$, $\JSL$ and $\Mod{\S}$ are precisely 
monoids, ordered monoids, idempotent semirings, and associative algebras over 
$\S$.
\end{expl}

\begin{proposition}[see \cite{ammu14}]
The free $\D$-monoid on a set $\Sigma$ is carried by $\Psi 
\Sigma^*\in \D$, the free algebra in $\D$ on the set $\Sigma^*$ of finite 
words 
over $\Sigma$. Its multiplication extends the concatenation of words in 
$\Sigma^*$, and 
its 
unit is the empty word $\epsilon$.
\end{proposition}

\begin{expl}
\begin{enumerate}[(1)]
\item In $\D=\Set$ or $\Pos$ we have $\Psi \Sigma^* = \Sigma^*$ (discretely 
ordered). 
\item In $\D=\JSL$ we have $\Psi\Sigma^* = \Pow_f 
\Sigma^*$, 
the idempotent semiring of all finite languages over $\Sigma$ w.r.t. union 
and  
concatenation of languages.

\item In $\D=\SMod{\S}$ we get $\Psi\Sigma^* = \S[\Sigma]$, the $\S$-algebra 
 of all polynomials $\Sigma_{i=1}^n c(w_i)w_i$ (equivalently, finite-support functions 
 $c\colon\Sigma^*\to \S$) w.r.t. the usual sum, scalar product and multiplication of polynomials.
\end{enumerate}
\end{expl}

\begin{rem}\label{rem:moncat}
Since the multiplication $\bullet\colon \under{M}\times\under{M}\to\under{M}$ of a 
$\D$-monoid $(M,1,\bullet)$ forms a bimorphism, it 
 corresponds to a morphism $\mu_M\colon M\t M\to M$ in $\D$, mapping $m\t 
 m'\in \under{M\t M}$ to $m\bullet m'\in \under{M}$. Likewise, the 
 unit 
$1\in\under{M}$ corresponds to the morphism $\iota_M\colon \one_\D \to M$ 
sending the generator of $\one_\D$ to $1$. We can thus represent 
a 
$\D$-monoid 
$(M,1,\bullet)$ as the 
triple $(M,\iota_M,\mu_M)$.
\end{rem}

\begin{rem}\label{rem:tensormonoid}
For any two $\D$-monoids $M$ and $N$, 
the tensor 
product 
$M\t N$ in $\D$ carries a $\D$-monoid structure with unit 
$\one_\D \xra{\cong} \one_\D\t \one_\D \xra{\iota_M\t \iota_N} M\t N$
and multiplication
$ (M\t N)\t (M\t N) \xra{\cong} (M\t M)\t 
(N\t N) \xra{\mu_M\t \mu_N} M\t N$,
see e.g. \cite{porst08}. Equivalently, the unit of $M\t N$ is the 
element 
$1_M\t 1_N$, and the multiplication is determined by $(m\t 
n)\bullet
(m'\t n') = (m\bullet_M m')\t (n\bullet_N n')$.
\end{rem}

\begin{defn}\label{def:monoidalfunc}
A \emph{monoidal functor} $(G,\theta)\colon\C\to\D$ is a functor $G\colon \C\to
\D$ with a morphism $\theta_1\colon \one_\D \to G\one_\C$ and morphisms
$\theta_{A,B}\colon GA\t GB \to G(A\t B)$ natural in $A,B\in \C$ such that the
following squares commute (omitting indices): \vspace{-0.4cm}
\begin{center}
\begin{tabular}[x]{cc}
\parbox{0.6\textwidth}{
$
\xymatrix@C-1em@R-1.5em{
(GA\t GB)\t GC \ar[r]^\alpha \ar[d]_{\theta\t GC} & GA \t 
(GB\t 
GC) \ar[d]^{GA\t \theta}\\
G(A\t B) \t GC \ar[d]_{\theta} & GA \t G(B\t C) 
\ar[d]^\theta \\
G((A\t B)\t C) \ar[r]_{G\alpha}& G(A\t (B\t C))
}
$
}
&
\parbox{0.4\textwidth}{
$
\xymatrix@R-1.5em{
GA\t \one_\D \ar[r]^-{GA\t \theta} \ar[d]_-{\rho} & GA\t G\one_\C 
\ar[d]^-\theta\\
GA & G(A\t \one_\C) \ar[l]_-{G\rho}
}
$\\
$\xymatrix@R-1.5em{
\one_\D\t GA \ar[r]^-{\theta\t GA} \ar[d]_-{\lambda} & G\one_\C\t GA 
\ar[d]^-\theta\\
GA & G(\one_\C\t A) \ar[l]_-{G\lambda}
}
$
}
\end{tabular}
\end{center}
Given another monoidal functor $(G',\theta')\colon \C\to \D$, a natural 
transformation $\phi\colon G\to G'$ is called \emph{monoidal} if the following 
diagrams commute:
\[  
\xymatrix@C+1em@R-1.5em{
GA\t GB \ar[r]^{\phi_A\t \phi_B} \ar[d]_{\theta} & G'A\t G'B 
\ar[d]^{\theta'}\\
G(A\t B) \ar[r]_{\phi_{A\t B}} & G'(A\t B)
}\qquad
\xymatrix@R-1.5em{
& \one_\D \ar[dl]_\theta \ar[dr]^{\theta'} & \\
G\one_\C \ar[rr]_{\phi_{\one_\C}} && G'\one_\C
}
\]
\end{defn}

\begin{expl}\label{ex:monoidalfunc}
\begin{enumerate}[(1)]
\item The functor 
$\under{\dash}\colon \D\to\Set$ is monoidal w.r.t. the universal map
$1\monoto \under{\one_\D}=\under{\Psi 1}$ and the bimorphisms
$t_{A,B}\colon  
\under{A}\times\under{B}\to \under{A\t B}$. Its left adjoint $\Psi\colon  
\Set\to\D$ 
is also monoidal: there 
is a natural \emph{isomorphism} $\theta_{X,Y}\colon  
\Psi X \t \Psi Y \cong \Psi(X\times Y)$ with $\theta_{X,Y}^{-1}(x,y)=x\t 
y$ for $(x,y)\in X\times Y$. Together with $\theta_1 = \id\colon  \one_\D\to 
\Psi 
1$, this makes $\Psi$ a monoidal functor.
\item In particular, the functors
$\under{\dash}\colon  \JSL\to\Set$ and $\Pow_f\colon  \Set\to\JSL$ (see 
Example \ref{ex:psi}(2)) are 
monoidal w.r.t. the morphisms chosen as in (1).
\item The forgetful functor $U\colon  \JSL\to\Pos$  has a left adjoint
  $\Dn_f\colon  
\Pos\to\JSL$ constructed as 
follows. For any 
poset $A$ and 
$X_0\seq A$ denote by
$\mathord{\dnarrow} X_0 := 
\{\,a\in A: 
a\leq x \text{ for some $x\in X_0$}\,\}$ the down-set generated by $X_0$. 
Then $\Dn_f$ maps a poset $A$ to 
$\Dn_f(A) \defeq \{\,X\seq A : X=\dnarrow X_0 \text{ for some finite $X_0\seq 
A$}\,\},
$ the semilattice (w.r.t. union) of finitely generated down-sets of $A$, and a 
monotone map $h\colon A\to B$ to 
the semilattice morphism
$\Dn_f(h) \colon \Dn_f(A)\to \Dn_f(B)$
with $\Dn_f(h)(X) = \dnarrow h[X]$. Both $U$ and $\Dn_f$ carry monoidal 
functors; the required morphisms, see Definition \ref{def:monoidalfunc}, are 
chosen in analogy to $\under{\dash}$ and $\Pow_f$ in (2).
\item As a trivial example, the identity functor $\Id\colon \D\to\D$ is monoidal 
w.r.t. the identity 
morphisms $\id\colon \one_\D\to \Id(\one_\D)$ and $\id\colon \Id(A)\t\Id(B)\to \Id(A\t 
B)$.
\end{enumerate}  
\end{expl}
The importance of monoidal functors is that they preserve monoid structures:

\begin{lemma}\label{lem:monlifting}
Let $(G,\theta)\colon \C\to \D$ be a monoidal functor. Then $G$ lifts to the 
functor $\ol G\colon \Mon{\C}\to 
\Mon{\D}$ mapping a $\C$-monoid $(M,\iota,\mu)$ to the $\D$-monoid 
\vspace{-0.2cm}
\[ (GM,\ \one_\D\xra{\theta} G\one_\C \xra{G\iota} GM,\ GM\t GM \xra{\theta} 
G(M\t 
M)\xra{G\mu} GM), \]
and a $\C$-monoid morphism $h$ to $Gh$. 
\end{lemma}

\begin{expl} 
\begin{enumerate}[(1)]
\item $\Pow_f\colon \Set\to\JSL$ lifts to the functor 
$\ol\Pow_f\colon \Mon{\Set}\to \Mon{\JSL}$ that maps a monoid $M$ to the  semiring 
$\ol\Pow_f M$ of finite subsets of $M$, with union as addition, and
multiplication 
$XY=\{\,xy: x\in Y,\,y\in Y\,\}$.
\item $\Dn_f\colon \Pos\to\JSL$ lifts to $\ol\Dn_f\colon 
\Mon{\Pos}\to\Mon{\JSL}$, mapping an ordered monoid $M$ to the semiring 
$\ol\Dn_f(M)$ of finitely generated down-sets of $M$, with
union as addition, and 
multiplication $XY = \dnarrow\{\,xy: x\in X,\,y\in Y\,\}$.
\end{enumerate}
\end{expl}

\begin{lemma}\label{lem:moncomp}
Let $(G,\theta)\colon \A\to \B$ and $(H,\sigma)\colon
\B\to\C$ be monoidal functors. Then the composite $HG\colon \A\to \C$ 
is a 
monoidal functor w.r.t. to
$H(\theta_1)\c \sigma_1\colon \one_\C\to HG(\one_\A)$ and $H(\theta_{A,B})\c 
\sigma_{GA,GB}\colon HGA\t HGB \to HG(A\t B)$.
\end{lemma}

\begin{defn}\label{def:monoidaladj}
A \emph{monoidal adjunction} between $\C$ and $\D$ is an adjunction 
$F\dashv U\colon \C\to\D$ such that $U$ and $F$ are monoidal 
functors and the unit $\eta\colon \Id_\D\to UF$ and counit $\epsilon\colon FU\to \Id_\C$ 
are monoidal 
natural 
transformations.
\end{defn}

\begin{expl}\label{ex:monoidaladj}
$\Id\dashv \Id\colon \D\to\D$, $\Dn_f\dashv U\colon \JSL\to\Pos$ and 
$\Psi\dashv 
\under{\dash}\colon\D\to\Set$ are monoidal adjunctions. We call the latter the 
\emph{monoidal adjunction of $\D$}.
\end{expl}

\begin{rem}
If $(H\dashv V\colon \C\to \B,\eta',\epsilon')$ and $(G\dashv 
U\colon \B\to \A,
\eta, \epsilon)$ are monoidal adjunctions, so is the composite adjunction 
$(HG\dashv UV\colon \C\to \A,\,U\eta' G\c \eta,\,\epsilon'\c H\epsilon V)$. 
Here $HG$ and $UV$ are the composites of Lemma \ref{lem:moncomp}.

\end{rem}


\begin{defn}
 A monoidal adjunction $F\dashv U\colon 
\C\to\D$ is called a \emph{concrete monoidal adjunction} if its composite with 
the 
monoidal adjunction of $\D$ is the monoidal adjunction of $\C$.
\end{defn}

\section{Languages and Algebraic Recognition}

In this section we set the scene for our categorical approach to 
Sch\"utzenberger products. For the rest of this paper let us fix a commutative 
variety $\D$ of algebras or ordered algebras, a commutative semiring 
$\S=(S,+,\o,0,1)$, and a concrete 
monoidal adjunction $F\dashv U\colon \Mod{\S}\to \D$ with unit $\eta\colon \Id\to UF$. 
Thus we have the 
diagram of 
functors below, where $\ol U$ and $\ol F$ are the lifted functors, see
Lemma \ref{lem:monlifting}, the
vertical functors are the forgetful functors, and $\Psi$ and $\S^{(\dash)}$ are the left 
adjoints to the forgetful functors of $\D$ and $\Mod{\S}$, see Example \ref{ex:psi}.
\vspace{-0.2cm}
\[
\xymatrix@R-1.5em{
\Alg{\S} \ar[d] \ar@<0.5ex>[rr]^{\ol{U}} && \Mon{\D} 
\ar@<0.5ex>[ll]^{\ol F} \ar[d] \\
 \Mod{\S} \ar@<0.5ex>[rr]^U \ar@<-1.0ex>[dr]_{\under{\dash}} & & \D 
 \ar@<0.5ex>[ll]^F  
\ar@<2.0ex>[dl]^{\under{\dash}} \\
& \Set \ar@<-1.0ex>[ur]^{\Psi} \ar[ul]_<<<<{\S^{(\dash)}}  & 
}
\]

\begin{expl}\label{ex:adjunctions}
In our applications we will choose the concrete monoidal adjunctions listed 
below. (The third and last column will be explained 
later.)
\small
\begin{center}
\begin{tabular}{|c|c|c|c|c|l|c|}
\hline
&$\S$ & $\C$ & $\D$ & $\xymatrix{\Mod{\S} 
\ar@<0.5ex>[r]^<<<<<{U} &
\D 
\ar@<0.5ex>[l]^<<<<<{F}}$ & $\D$-monoids & $M\diamond N$ 
carried by\\
\hline
1&$\{0,1\}$ &$\BA$ & $\Set$ & $\xymatrix{\JSL 
\ar@<0.5ex>[r]^<<<<<{\under{\dash}} &
\Set 
\ar@<0.5ex>[l]^<<<<<{\Pow_f}}$ & monoids & $M\times \Pow_f(M\times 
N)\times N$ \\
2&$\{0,1\}$ & $\DL$ & $\Pos$ & $\xymatrix{\JSL \ar@<0.5ex>[r]^<<<<<U & \Pos 
\ar@<0.5ex>[l]^<<<<<{\Dn_f}}$  & ord. monoids & $M\times \Dn_f(M\times 
N)\times N$ \\
3&$\{0,1\}$ & $\JSL$ & $\JSL$ & $\xymatrix{\JSL 
\ar@<0.5ex>[r]^<<<<<{\Id} & \JSL
\ar@<0.5ex>[l]^<<<<<{\Id}}$  & id. semirings & $M\times (M\ast N)\times 
N$\\
4&$\K$ & $\Vect{\K}$ &  $\Vect{\K}$ & $\xymatrix{\Vect{\K} 
\ar@<0.5ex>[r]^<<<<<{\Id} & \Vect{\K} 
\ar@<0.5ex>[l]^<<<<<{\Id}}$  & $\K$-algebras & $M\times (M\t N)\t N$\\
5&$\S$ & ? & $\Mod{\S}$ & $\xymatrix{\Mod{\S} 
\ar@<0.5ex>[r]^<<<<<{\Id} & \Mod{\S} 
\ar@<0.5ex>[l]^<<<<<{\Id}}$  & $\S$-algebras & $M\times (M\ast N)\t N$\\
\hline
\end{tabular}
\end{center}
\normalsize
\end{expl}

\begin{notation}
We can view the semiring $\S$ as (i) an $\S$-algebra $\S_{\mathbf{Alg}}\in
\Alg{\S}$ with scalar product given by the multiplication of $\S$, (ii) a
$\D$-monoid $\S_{\mathbf{Mon}}\in \Mon{\D}$ (by applying $\ol U$ to
$\S_{\mathbf{Alg}}$), (iii) an $\S$-module $\S_{\mathbf{Mod}}\in \Mod{\S}$ (by
applying the forgetful functor to $\S_{\mathbf{Alg}}$) and (iv) an object
$\S_\D$ of $\D$ (by applying $U$ to $\S_{\mathbf{Mod}}$). The $\D$-monoid
$\S_\mathbf{Mon}$ is carried by the object $\S_\D$, and its multiplication is a
morphism of $\D$ that we denote by $\sigma\colon \S_\D \t \S_\D \to \S_\D$. For ease of notation we will usually drop the indices and simply write $\S$ for $\S_\D$, $\S_{\mathbf{Mod}}$, etc.
\end{notation}

\begin{defn}\label{def:recog}
\begin{enumerate}[(1)]
\item A \emph{language} (a.k.a. a 
\emph{formal power series}) over a finite alphabet $\Sigma$ is a map $L\colon 
\Sigma^*\to S$. Denote by $L_\D\colon \Psi\Sigma^*\to \S$ the adjoint 
transpose 
of 
$L$ w.r.t. the adjunction $\Psi\dashv \under{\dash}\colon \D\to\Set$.
 A
$\D$-monoid morphism $f\colon \Psi\Sigma^*\to M$ \emph{recognizes} $L$ if there 
is a morphism $p\colon M\to \S$ in 
$\D$ with $L_\D =p\c f$. In this case, we also say that $M$ \emph{recognizes} 
$L$ (via 
$f$ and $p$).
\item The \emph{marked Cauchy product} of two languages $K,L\colon \Sigma^*\to S$ 
w.r.t. a letter $a\in\Sigma$ is the language 
$KaL\colon \Sigma^*\to S$ with $(KaL)(u) = \sum_{u=vaw} K(v)\o L(w)$.
\end{enumerate}
\end{defn}
For $\S=\{0,1\}$, a language $L\colon \Sigma^*\to \{0,1\}$ corresponds to a classical language $L\seq \Sigma^*$ by taking the preimage of $1$. Under this identification, we have
$KaL = \{\,vaw: v\in K,\,w\in L\,\}$. Our concept of language recognition by 
$\D$-monoids originates in \cite{amu15} and specializes to several related 
notions 
from 
the literature:

\begin{expl}\label{ex:monrec}
\begin{enumerate}[(1)]
\item $\D = \Set$ with $\S=\{0,1\}$:  a map $p\colon M\to \{0,1\}$ 
corresponds 
to a 
subset $p^{-1}[1]\seq M$. Thus a monoid 
morphism $f\colon \Sigma^*\ra M$ recognizes the language $L\seq \Sigma^*$ iff $L$ 
is the 
preimage under $f$ of some subset of $M$. This is the classical notion of 
language recognition by a monoid, see e.g.\ \cite{pin15}. 
\item $\D = \Pos$ with $\S_\Pos=\{0<1\}$: given an ordered 
monoid $M$, a monotone map $p\colon M\ra\{0,1\}$ defines an upper set 
$p^{-1}[1]\seq M$. Hence a monoid morphism $f\colon \Sigma^*\ra M$ recognizes 
$L\seq \Sigma^*$ iff 
$L$ is the preimage under $f$ of some upper set of $M$. This notion of 
recognition is 
due to Pin~\cite{pin95}.
\item $\D=\JSL$ with $\S_\JSL=\{0<1\}$: for 
any idempotent semiring $M$, a semilattice morphism $p\colon M\ra \{0,1\}$ 
defines an ideal $I=p^{-1}[0]$, i.e. 
a nonempty down-set closed under joins. Hence a language $L\seq \Sigma$ is 
recognized by a 
semiring morphism $f\colon \Pow_f \Sigma^* \ra M$ via $p$ iff $\ol L = 
\Sigma^*\cap 
f^{-1}[I]$. Here we identify $\Sigma^*$ with the set of all singleton languages 
$\{w\}$, $w\in \Sigma^*$. This is the concept of language recognition by 
idempotent semirings introduced 
by Pol\'ak \cite{polak01}.
\item $\D=\Mod{\S}$: given an $\S$-algebra $M$, a formal power series $L\colon 
\Sigma^*\to S$ is 
recognized by 
$f\colon \S[\Sigma]\to 
M$ 
via $p\colon M\ra \S$ iff $L_{\Mod{\S}}= p\c f$. This notion of recognition is due 
to Reutenauer~\cite{reu80}. If $\S$ is a commutative ring, the 
power series 
recognizable 
by $\S$-algebras of finite type (i.e. $\S$-algebras whose underlying 
$\S$-module is finitely generated) are precisely rational power series.
\end{enumerate}
\end{expl}

\section{The Sch\"utzenberger Product}
We are ready to introduce the Sch\"utzenberger product for 
$\D$-monoids. Fix two $\D$-monoids
$(M,1,\bullet)$ and $(N,1,\bullet)$, and write $xy$ for $x\bullet y$. 
Our goal is to construct a $\D$-monoid $M\diamond N$ that recognizes 
all 
marked products of languages recognized by $M$ and $N$, and is the 
``smallest'' such $\D$-monoid 
(Theorems \ref{thm:schurec}, \ref{thm:universalprop}, 
\ref{thm:schurec2}).
\vspace{0.05cm}
\begin{construction} As a preliminary step, we define a 
$\D$-monoid $M\ast N$ as fol-
\noindent\parbox{0.77\textwidth}{
lows. Call a family $\{\,f_i\colon A\to B_i\,\}_{i\in I}$ in 
$\D$ \emph{separating} if the morphism $f \colon A\to \prod_i B_i$ with $f(a) 
= (f_i(a))_{i\in I}$ is 
injective (resp.
order-reflecting when $\D$ is a variety of ordered algebras).  Any family 
$\{f_i\}$ yields a separating family $\{\,f_i'\colon A'\to B_i\,\}_{i\in 
I}$ 
by 
factorizing
$f = m\c \pi$ with $\pi$ surjective and $m$ injective
(resp. order-reflecting), and setting $f_i' \defeq p_i\c m$, where $p_i$ is  
\vspace{0.05cm}
}
\parbox{0.23\textwidth}{
  $
  \xymatrix@R-1.5em@C-1.5em{
  A \ar@/^4ex/[rr]^{f_i}\ar@{->>}[r]^\pi \ar[dr]_{f}  & A' 
  \ar[r]^{f_i'} 
  \ar@{>->}[d]^-{m} & B_i\\
  & \prod_i B_i \ar[ur]_{p_i} &}$ 
}
the projection. Now consider the family of all morphisms 
$\sigma\c (p\t q)\colon M\t N\to \S$, where $p\colon M\to\S$ and $q\colon 
N\to \S$ are arbitrary morphisms in $\D$. Applying the above construction to 
this family $\{\,\sigma\c (p\t q)\,\}_{p,q}$ gives an algebra $M\ast N$ in 
$\D$, a surjective morphism $\pi\colon M\t 
N\epito M\ast 
N$, and a separating family $\{\,p\ast 
q\colon M\ast N\to \S\,\}_{p,q}$, making the following diagram commute for all 
$p$ and $q$:
\begin{equation}\label{eq:s1s2}
\begin{gathered}
\xymatrix@R-2.0em{
& \S\otimes \S \ar[dr]^\sigma &\\
M\otimes N \ar[ur]^{p\otimes q} \ar@{->>}[r]_\pi  & {M\ast 
N} 
\ar[r]_<<<<<{p\ast q}  & \S\\
}
\end{gathered}
\end{equation}
\end{construction}

\begin{notation} For any $m\in \under{M}$ and $n\in\under{N}$, we write $m\ast 
n$ for the element
$\pi(m\t n)\in\under{M\ast N}$.
\end{notation}

\begin{lemma}\label{lem:mastn}
There exists a (unique) $\D$-monoid structure on $M\ast N$ such that $\pi: M\t 
N\epito M\ast N$ is a 
$\D$-monoid morphism. The multiplication is determined by $(m\ast n)\bullet 
(m'\ast n') = (mm')\ast (nn')$, and the unit is $1\ast 1$.
\end{lemma}

\begin{expl}\label{ex:jointlym}
 For $\D=\Set$, $\Pos$ or $\Vect{\K}$, the family 
 $\{\,\sigma\c (p\t q)\,\}_{p,q}$ is 
already
separating, and therefore $M\ast N = M\t N$ and $p\ast q = \sigma\c 
(p\t q)$. For $\D=\JSL$ and in case $M$ and $N$ are \emph{finite} idempotent 
semirings, we can 
describe the idempotent semiring $M\ast N$ as 
follows. For any subset $X\seq M\times N$, let $[X]\seq M\times N$ consist of 
those elements $(m,n)\in M\times N$ such that, for all ideals $I\seq M$ and 
$J\seq N$ with  $m\not\in I$ and  $n\not\in 
J$, there exists some $(x,y)\in X$ with $x\not\in I$ and 
$y\not\in J$. This gives us the closure operator $X\mapsto [X]$ on the power set
of $M\times N$ in~\cite{klimapolak10}. One can show that $M\ast N$ is isomorphic to the idempotent semiring
of all closed 
subsets of $M\times N$, with sum and product defined by $[X]\vee[Y] = [X\cup 
Y]$ and
$[X][Y]=[XY]$, where $XY=\{\,xy: x\in Y,\,y\in Y\,\}$.
\end{expl}

\begin{defn}
The \emph{Sch\"utzenberger product} of $M$ and $N$ is the $\D$-monoid 
$M\diamond N$ carried 
by the product $M\times UF(M\ast N)\times N$ in $\D$ and equipped with the 
following monoid structure: representing elements $(m,a,n)\in 
\under{M}\times\under{F(M\ast N)}\times \under{N}$ as upper triangular matrices
$
 \begin{pmatrix}
 m & a\\
 0 &  n
 \end{pmatrix},
$
the multiplication and unit are given by
 \[
 \begin{pmatrix}
 m & a\\
 0 &  n
 \end{pmatrix}
 \begin{pmatrix}
 m' & a'\\
 0 &  n'
 \end{pmatrix}
 =
 \begin{pmatrix}
 mm'~~ & \eta(m\ast 1)\o a'+a\o \eta (1\ast n')\\
 0 &  nn'
 \end{pmatrix}
\quad\text{and}\quad
 \begin{pmatrix}
 1 & 0\\
 0 &  1
 \end{pmatrix}.
 \]
 Here $\eta\colon M\ast N\to UF(M\ast N)$ is the universal map, and the sum,
 product and $0$ in the upper right components are taken in the $\S$-algebra $\ol 
 F(M\ast N)$.
\end{defn}

\begin{lemma}\label{lem:schuetzwelldef}
$M\diamond N$ is a well-defined $\D$-monoid, and the product projections 
$\pi_M\colon 
M\diamond N\to M$ and $\pi_N\colon M\diamond N \to N$ are $\D$-monoid morphisms.
\end{lemma}

\begin{expl}
For the categories and adjunctions of Example \ref{ex:adjunctions}, we recover 
four notions of Sch\"utzenberger products known in the literature, and obtain 
a new Sch\"utzenberger product for $\S$-algebras:
\begin{enumerate}[(1)]
\item $\D=\Set$: given monoids $M$ and $N$, the monoid $M\diamond N$ 
is carried by the set $M\times \Pow_f (M\times 
N)\times 
N$, with multiplication and unit
\[
 \begin{pmatrix}
 m & X\\
 0 &  n
 \end{pmatrix}
 \begin{pmatrix}
 m' & X'\\
 0 &  n'
 \end{pmatrix}
 =
 \begin{pmatrix}
 mm'~~ & mX'\cup Xn'\\
 0 &  nn'
 \end{pmatrix}
 \quad\text{and}\quad
\begin{pmatrix}
1 & \emptyset\\
0 &  1
\end{pmatrix},
 \]
where $mX' = \{\,(my,z): (y,z)\in X'\,\}$ and $Xn' = \{\,(y,zn'): 
(y,z)\in X\,\}$. This is 
the original construction of Sch\"utzenberger \cite{sch65}.
\item $\D=\Pos$: for ordered monoids $M$ and $N$, the ordered monoid 
$M\diamond N$ is carried by the poset $M\times \Dn_f(M\times 
N)\times N$ with multiplication and unit
\[
 \begin{pmatrix}
 m & X\\
 0 &  n
 \end{pmatrix}
 \begin{pmatrix}
 m' & X'\\
 0 &  n'
 \end{pmatrix}
 =
 \begin{pmatrix}
 mm'~~ & \dnarrow(mX'\cup Xn')\\
 0 &  nn'
 \end{pmatrix}
 \quad\text{and}\quad
\begin{pmatrix}
1 & \emptyset\\
0 &  1
\end{pmatrix}.
 \]
This construction is due to Pin \cite{pin03}.
\item $\D=\JSL$: given idempotent semirings $M$ and $N$, the idempotent 
semiring $M\diamond N$ is carried by the semilattice $M\times (M\ast 
N)\times N$. If $M$ and $N$ are finite, $M\ast N$ is the 
idempotent semiring of closed subsets of 
$M\times N$ by Example \ref{ex:jointlym}, and the multiplication and unit of 
$M\diamond N$ are given by 
\[
 \begin{pmatrix}
 m & X\\
 0 &  n
 \end{pmatrix}
 \begin{pmatrix}
 m' & X'\\
 0 &  n'
 \end{pmatrix}
 =
 \begin{pmatrix}
 mm'~~ & [mX'\cup Xn']\\
 0 &  nn'
 \end{pmatrix}
 \quad\text{and}\quad
\begin{pmatrix}
1 & \emptyset\\
0 &  1
\end{pmatrix}.
 \]
For the finite case, this construction is due to Kl\'ima and Pol\'ak 
\cite{klimapolak10}.
\item $\D=\Vect{\K}$: given $\K$-algebras $M$ and $N$, the $\K$-algebra 
$M\diamond N$ is carried by the vector space $M\times (M\t N)\times 
N$ with multiplication and unit
\[
 \begin{pmatrix}
 m & z\\
 0 &  n
 \end{pmatrix}
 \begin{pmatrix}
 m' & z'\\
 0 &  n'
 \end{pmatrix}
 =
 \begin{pmatrix}
 mm'~~ & mz'+zn'\\
 0 &  nn'
 \end{pmatrix}
 \quad\text{and}\quad
\begin{pmatrix}
1~ & ~0\t 0\\
0~ &  ~1
\end{pmatrix},
 \]
 where $mz' = (mm_0)\t n_0$ for $z'=m_0\t n_0$, and extending via bilinearity 
 for 
 arbitrary $z$; similarly for $zn'$. 
This construction is due to Reutenauer \cite{reu80}.
\item $\D=\Mod{\S}$: given $\S$-algebras $M$ and $N$, the $\S$-algebra 
$M\diamond N$ is carried by the $\S$-module $M\times (M\ast N)\times N$ with 
multiplication and unit
\[
 \begin{pmatrix}
 m & z\\
 0 &  n
 \end{pmatrix}
 \begin{pmatrix}
 m' & z'\\
 0 &  n'
 \end{pmatrix}
 =
 \begin{pmatrix}
 mm'~~ & mz'+zn'\\
 0 &  nn'
 \end{pmatrix}
 \quad\text{and}\quad
\begin{pmatrix}
1~ & ~0\ast 0\\
0~ & ~1
\end{pmatrix},
 \]
 where $mz' = (mm_0)\ast n_0$ for $z'=m_0\ast n_0$, and similarly for $zn'$. 
   This example specializes to (3) and (4) by taking $\S=\{0,1\}$ and $\S=\K$, respectively, 
   but appears to be new construction for other semirings $\S$.
\end{enumerate}
\end{expl}
The following theorem gives the key property of $M\diamond N$.
\begin{theorem}\label{thm:schurec}
Let $K,L\colon \Sigma^*\to S$ be languages recognized by $M$ and $N$,
respectively. Then 
$M\diamond N$ 
recognizes the languages $K$, $L$ and $K a 
L$ for all $a\in\Sigma$.
\end{theorem}
Next, we aim to show that $M\diamond 
N$ is the ``smallest'' 
$\D$-monoid satisfying the statement of the above theorem. This 
requires 
further assumptions 
on 
our setting.

\begin{notation}
Recall from \eqref{eq:s1s2} the morphism $p\ast q\colon M\ast N\to \S$. We denote 
its adjoint transpose w.r.t. the adjunction $F\dashv U$ by  $\ol{p\ast 
q}\colon F(M\ast N)\to \S$.
\end{notation}

\begin{assumptions}\label{asm:dual}
From now on, suppose that:
\begin{enumerate}[(i)]
\item $\D$ is locally finite, i.e. every finitely generated algebra of $\D$ is 
finite.
\item Epimorphisms in $\D$ and $\Mod{\S}$ are surjective.
\item $\D(M,\S)$, $\D(N,\S)$, and 
$\{\,U(\ol{p\ast q})\colon UF(M\ast N)\to \S\,\}_{p\colon M\to 
\S,\,q\colon N\to \S}$
 are separating families of morphisms in $\D$.
\item There is a locally finite variety $\C$ of algebras such that
the full subcategories $\C_f$ and $\D_f$ on \emph{finite} algebras are dually 
equivalent. We denote the equivalence functor by $E\colon 
\D_f^{op}\simeq\C_f$.
\item The semiring $\S$ is finite, and $E(\S)\cong \one_\C$.
\end{enumerate}
\end{assumptions}
Let us indicate the intuition behind our assumptions.
First, (i) and (ii) imply that $M\diamond N$ 
is finite 
if $M$ 
and $N$ are. This is important, as one is usually interested in 
language 
recognition by \emph{finite} $\D$-monoids. (iii) expresses that the semiring $\S$ 
has enough structure to separate elements of $M$, $N$ and $UF(M\ast N)$, 
the three components of the Sch\"utzenberger product $M\diamond N$, by suitable
morphisms into $\S$. This 
 technical condition on $\S$ is the crucial 
ingredient 
for proving the ``smallness'' of $M\diamond N$ (Theorem 
\ref{thm:universalprop}). Finally, the variety $\C$ in (iv) and (v) 
 will be used to determine, via duality, the algebraic operations to express 
 languages 
recognized by $M\diamond N$ in terms of languages recognized by $M$ and $N$ 
(Theorem \ref{thm:schurec2}).

\begin{expl}\label{ex:predualcats} The categories and adjunctions 
of Example \ref{ex:adjunctions}(1)-(4) satisfy our assumptions. Here we 
briefly 
sketch the dualities; see \cite{ammu14,ammu15} for details.
\begin{enumerate}[(1)]
\item For $\D=\Set$, choose $\C=\BA$ (boolean algebras). Stone duality 
\cite{Johnstone1982} gives a dual equivalence $E\colon \Set_f^{op}\simeq \BA_f$ 
mapping a finite set to the boolean 
algebra of all subsets.
\item For $\D=\Pos$, choose $\C = \DL$ (distributive 
lattices with $0,1$). Birkhoff duality \cite{birkhoff37} gives a dual 
equivalence $E\colon \Pos_f^{op}\simeq\DL_f$ mapping a finite poset to 
the 
lattice of all down-sets.
\item For $\D=\JSL$, choose $\C = \JSL$. The 
dual equivalence $E\colon \JSL_f^{op}\simeq\JSL_f$ maps a finite semilattice 
$(X,\vee)$ to its opposite 
semilattice $(X,\wedge)$, see \cite{Johnstone1982}. 
\item For $\D=\Vect{\K}$, $\K$ a \emph{finite} field, choose 
$\C=\Vect{\K}$. The dual equivalence $E\colon 
\Vect{\K}_f\simeq\Vect{\K}_f^{op}$ 
maps
a space $X$ to its dual space $X^*=\hom(X,\K)$.
\end{enumerate}
\end{expl}

\begin{notation}\label{not:lmnf}
For any $\D$-monoid morphism $f\colon \Psi\Sigma^* \to M\diamond N$, put
\vspace{-0.2cm}
\[
  \L_{M,N}(f) := \{\,K, L, KaL \mid \text{$a\in\Sigma$, $\pi_M\c f$ recognizes 
  $K$,  $\pi_N\c
    f$ recognizes $L$}\,\} 
\]
\end{notation}

\begin{theorem}\label{thm:universalprop}
Let $f\colon \Psi\Sigma^*\to M\diamond N$ and $e\colon \Psi\Sigma^*\epito 
P$ be two $\D$-monoid morphisms. If $e$ is surjective and recognizes all 
languages in $\L_{M,N}(f)$, then there exists a 
unique
$\D$-monoid morphism $h\colon P\to M\diamond N$ with $h\c e = f$.
\end{theorem}
Using our duality framework, this theorem can be rephrased in terms of 
language operations.
Recall that $E(\S)\cong\one_\C$ by Assumption \ref{asm:dual}(v). Putting $O_\C 
\defeq
E(\one_\D)$, we obtain a bijection
$i\colon S \cong \D(\one_\D, \S) \cong \C(E(\S),E(\one_\D)) \cong \C(\one_\C, 
O_\C) \cong 
\under{O_\C}$.

\begin{defn}
For any $n$-ary operation symbol $\gamma$ in the signature of $\C$ and 
languages $L_1,\ldots, L_n\colon \Sigma^*\to S$, the language 
$\ol\gamma(L_1,\ldots,L_n)\colon \Sigma^*\to S$ is given by
$\ol\gamma(L_1,\ldots,L_n)(u) \defeq i^{-1}(\,\gamma^{O_\C}(\,i(L_1u), \ldots, i(L_n u)\,)\,)$.
The operations $\ol\gamma$ are called the \emph{$\C$-algebraic operations} on 
the set of languages 
over $\Sigma$.
\end{defn}

\begin{expl} $O_\BA\cong\{0,1\}$ is the two-element boolean algebra, and the 
$\BA$-algebraic operations are precisely the boolean operations (union, 
intersection, complement, $\emptyset$, $\Sigma^*$) on languages. For example, 
the operation symbol $\vee$ induces the language operation $(K\,\ol\vee\,L)(u) = 
K(u)\vee L(v)$ corresponding to the union of languages.  Similarly, for $\C=\DL$ we get 
union, intersection, $\emptyset$, $\Sigma^*$, for $\C=\JSL$ we get union and 
$\emptyset$, and for $\C=\Vect{\K}$ we get sum, scalar product and $\emptyset$.
\end{expl}
All our constructions and results so far apply to arbitrary $\D$-monoids. 
However, in the following theorem we need to restrict to \emph{finite} 
$\D$-monoids. Recall that the \emph{derivatives} of a language 
$L\colon\Sigma^*\to S$ 
are the languages $a^{-1}L,\,L{a^{-1}}\colon \Sigma^*\to S$ (where 
$a\in\Sigma$) 
defined by 
$(a^{-1}L)(u) = L(au)$ and $(La^{-1})(u) = L(ua)$.

\begin{theorem}\label{thm:schurec2}
Let $M$ and $N$ be finite $\D$-monoids and $f\colon \Psi\Sigma^*\to M\diamond 
N$ be a $\D$-monoid morphism.  Then every 
language
recognized by $f$ lies in the closure of $\L_{M,N}(f)$ 
under the $\C$-algebraic operations and derivatives.
\end{theorem}
Our proof uses the \emph{Local Variety Theorem} of \cite{ammu14}: for any 
finite set 
 $\V$ of recognizable languages closed under $\C$-algebraic 
operations and derivatives, there is a finite $\D$-monoid recognizing 
precisely the languages of $\V$.
Coincidentally, for each of our categories of Example 
\ref{ex:adjunctions}(1)-(4) it suffices to take the closure of 
$\L_{M,N}(f)$ 
under $\C$-algebraic operations, as this set is already derivative-closed. For 
example, for $\C=\Vect{\K}$ we have 
$a^{-1}(KaL)=(a^{-1}K)aL+K(\epsilon)L$, i.e. $a^{-1}(KaL)$ is a linear 
combination of 
languages in $\L_{M,N}(f)$ and thus lies in the closure of $\L_{M,N}(f)$
under $\Vect{\K}$-operations. For $\D=\Set$, $\Pos$ and $\JSL$, Theorem 
\ref{thm:schurec2} 
then gives

\begin{corollary}[Reutenauer \cite{reutenauer79}, Pin \cite{pin03}, Kl\'ima and Pol\'ak \cite{klimapolak10}]
Let $M$ and $N$ be finite monoids [ordered monoids, idempotent semirings]. 
Then any language recognized by the Sch\"utzenberger product $M\diamond N$ is 
a boolean combination [positive boolean combination, finite union] of 
languages of the form $K$, $L$ and $KaL$, where $K$ is recognized by $M$, $L$ 
is recognized by $N$, and $a\in\Sigma$.
\end{corollary}
For $\D=\Vect{\K}$, we obtain a new result for formal power series:

\begin{corollary}
Let $M$ and $N$ be finite algebras over a finite field $\K$. Then any language 
recognized by 
$M\diamond N$ is a linear combination of power series of the form $K$, 
$L$ and $KaL$, where $K$ is recognized by $M$, $L$ is recognized by $N$, and 
$a\in\Sigma$.
\end{corollary}

\section{Conclusions and Future Work}
We presented a uniform approach to Sch\"utzenberger products 
for various algebraic structures. Our categorical framework encompasses all 
known instances of Sch\"utzenberger products in the 
setting 
of regular 
languages. Two related constructions are the 
Sch\"utzenberger products for \emph{$\omega$-semigroups} \cite{carton93} 
(dealing with 
$\infty$-languages), and for \emph{boolean spaces with 
internal 
monoids} \cite{gpr16} (dealing with non-regular  
languages). Neither of 
these structures are monoids in the categorical sense, and thus are not 
covered by our present setting. The use of monads as in 
\cite{boj15,camu16,uacm16} might pave the way to extending the scope of our 
work. 

Since our main focus in the present paper was to establish the categorical 
setting, we restricted to \emph{binary} Sch\"utzenberger products $M\diamond 
N$. For (ordered) monoids and semirings, a non-trivial $n$-ary generalization 
of the Sch\"utzenberger product
is known \cite{straubing81,pin03,polak01}, and we aim to adapt our results to arbitrary $n$.

%% file: appendix.tex
\noindent This appendix provides all omitted
proofs, as well as additional details for our examples. We start with a 
review of concepts from universal 
algebra and 
category theory.

\section{Algebraic Toolkit}

\begin{appsec}\label{app:varieties}\textbf{Varieties of algebras.}
Fix a finitary signature~$\Gamma$, i.e. a set of
operation symbols with finite arities. A 
\emph{$\Gamma$-algebra} is a 
set $A$
equipped with an operation $\gamma^A: A^n\to A$ for each $n$-ary 
$\gamma\in\Gamma$, and a \emph{morphism} of $\Gamma$-algebras 
is a map preserving these operations. \emph{Quotients} and 
\emph{subalgebras} of $\Gamma$-algebras are represented by 
surjective resp. injective morphisms. A \emph{variety of algebras} 
is 
a class of $\Gamma$-algebras closure under quotients, subalgebras, 
and products. Equivalently, a variety is class of $\Gamma$-algebras 
specified by equations $s=t$ between 
$\Gamma$-terms.
\end{appsec}

\begin{appsec}\label{app:orderedvarieties}\textbf{Varieties of ordered algebras.}
An \emph{ordered $\Gamma$-algebra} is a 
poset $A$ equipped with a monotone operation $\gamma^A: A^n\to A$ for each 
$n$-ary $\gamma\in\Gamma$, and a \emph{morphism} of 
ordered $\Gamma$-algebras is a monotone map preserving these operations. 
\emph{Quotients} of ordered algebras are 
represented by 
surjective morphisms, and \emph{subalgebras} by order-reflecting morphisms $m$ 
(i.e. $mx\leq my$ iff $x\leq y$). A \emph{variety of ordered 
algebras} 
is 
a class of ordered $\Gamma$-algebras closed under quotients, subalgebras, 
and products. Equivalently, a variety is class of ordered $\Gamma$-algebras 
specified by inequations $s\leq t$ between 
$\Gamma$-terms. 
\end{appsec}
In the following, let $\D$ always denote a variety of algebras or ordered algebras.

\begin{appsec}\label{app:commvarieties}\textbf{Commutative varieties.}
A variety $\D$ of algebras or ordered algebras is \emph{commutative} if, for 
any 
$A\in\D$ and any 
$n$-operation 
symbol $\gamma\in\Gamma$, the corresponding operation 
$\gamma^A: \under{A}^n\to \under{A}$ carries a morphism $\gamma^A: A^n\to A$ 
of $\D$. Equivalently, for any two 
algebras $A,B\in \D$, the set $[A,B]$ of morphisms from $A$ to $B$ forms an 
algebra of $\D$ under the pointwise $\Gamma$-operations, i.e. 
$[A,B]$ carries a 
subalgebra of $B^{\under{A}}$, the $\under{A}$-fold power of $B$. 
\end{appsec}

\begin{appsec}\label{app:congruences}\textbf{Congruences and stable preorders.} 
\begin{enumerate}[(1)]
\item A \emph{congruence} on a $\Gamma$-algebra $A$ is an 
equivalence relation $\equiv$ on $A$ such that for all $n$-ary operations 
$\gamma\in\Gamma$ and elements $a_1,\ldots,a_n,b_1,\ldots,b_n\in A$, 
\[a_i\equiv b_i\,\, (i=1,\ldots,n) \quad\text{implies}\quad 
\gamma^A(a_1,\ldots,a_n)\equiv\gamma^A(b_1,\ldots, b_n).\]
The set $A/\mathord{\equiv}$ of equivalence classes carries a 
$\Gamma$-algebra structure defined by
\[ \gamma^{A/\mathord{\equiv}}([a_1],\ldots,[a_n]) := 
[\gamma^A(a_1,\ldots,a_n)],\]
and the projection map $\pi: A\epito A/\mathord{\equiv}$, $a\mapsto [a]$, is a 
surjective morphism of $\Gamma$-algebras.
\item Let $(A,\leq)$ be an ordered $\Gamma$-algebra. A \emph{stable preorder}  
on $A$ is a preorder $\preceq$ on $A$ such that
(i) $a\leq b$ implies $a\preceq b$, and (ii) for all $n$-ary operations 
$\gamma\in\Gamma$ and elements $a_1,\ldots,a_n,b_1,\ldots,b_n\in A$, 
\[a_i\preceq b_i\,\, (i=1,\ldots,n) \quad\text{implies}\quad 
\gamma^A(a_1,\ldots,a_n)\preceq\gamma^A(b_1,\ldots, b_n).\]
For any stable preorder, the equivalence relation $\equiv\,:=\,\preceq \cap 
\succeq$ forms a congruence on $A$ in the sense of (1), and 
$A/\mathord{\equiv}$ becomes an ordered $\Gamma$-algebra by setting 
$[a]\leq[a']$ iff $a\leq a'$. We write $A/\mathord{\preceq}$ for this ordered 
algebra. The projection map $\pi: A\epito A/\preceq$, $a\mapsto [a]$, 
is a 
surjective morphism of ordered $\Gamma$-algebras.
\end{enumerate}
\end{appsec}

\begin{appsec}\label{app:jointlyinj}\textbf{Separating families.}
A family $\{f_i: A\to B_i\}_{i\in I}$ of morphisms in $\D$ is 
\emph{separating} if the 
  morphism 
  $f: A\to \prod_i B_i$ with $f(a) = 
  (f_i(a))_{i\in I}$ is injective (resp. order-reflecting if $\D$ is a variety 
  of ordered algebras). Equivalently, for any 
  two elements $a,a'\in A$ with $a\neq a'$ (resp. $a\not\leq a'$), there 
  exists an $i\in I$ with $f_i(a)\neq f_i(a')$ (resp. $f_i(a)\not\leq 
  f_i(a')$).
  Suppose that, for each $i\in I$, another separating family $\{\,g_{i,j}: 
  B_i\to C_{i,j}\,\}_{j\in J_i}$ is given. Then the combined family 
  $\{\,g_{i,j}\c f_i: A\to C_{i,j}\,\}_{i\in I, j\in J_i}$ is also separating.
\end{appsec}

\begin{appsec}\label{app:factsystems}\textbf{Factorization systems.}
Any variety $\D$ of algebras or ordered algebras has the factorization system 
of surjective 
and injective (resp. order-reflecting) morphisms. This means that (i) any 
morphism $h: A\to B$ has a factorization $h=m\c e$ with $e$ surjective and $m$ 
injective (resp. order-reflecting), and (ii) the \emph{diagonal fill-in} 
property holds: given a commutative square as displayed below with $e$ surjective and 
$m$ injective (order-reflecting), there is a unique morphism $d$ making both 
triangles commutative: 
\[
\xymatrix{
D \ar@{->>}[r]^e \ar[d]_g & C \ar[d]^{h} \ar@{-->}[dl]_{d}\\
A \ar[r]_{m} & B_i
}
\]  
The diagonal fill-in property generalizes to families of morphisms: suppose 
that $e$, $g$, 
$h_i$ and 
$m_i$ ($i\in I$) are morphisms with $h_i\c e=m_i\c g$ for all $i$. 
If $e$ is surjective and the family $\{\,m_i\,\}_{i\in I}$ is separating, then there exists a unique morphism 
$d$ with $d\c e = g$ and $m_i\c d=h_i$ for all $i$:
\[
\xymatrix{
D \ar@{->>}[r]^e \ar[d]_g & C \ar[d]^{h_i} \ar@{-->}[dl]_{d}\\
A \ar[r]_{m_i} & B_i
}
\] 
\end{appsec}

\begin{appsec}\label{app:tensorproducts}\textbf{Tensor products.}
Let $\D$ be \emph{commutative} variety of algebras or ordered algebras
\begin{enumerate}[(1)]
\item Let $A,B,C\in \D$. By a 
\emph{bimorphism} from 
$A$, $B$ to 
$C$ is meant a function $f: \under{A}\times\under{B}\ra\under{C}$ such that 
the maps 
$f(a,\mathord{-}): \under{B}\ra \under{C}$ and $f(\mathord{-},b): \under{A}\ra 
\under{C}$ carry morphisms of $\D$ for every $a\in\under{A}$ and $b\in 
\under{B}$. A \emph{tensor product} of $A$ and $B$ is a universal bimorphism 
$t_{A,B}: \under{A}\times \under{B} \ra \under{A\t B}$, i.e.
for any bimorphism $f: \under{A}\times \under{B} \ra \under{C}$
there is a 
unique  $f': A\t B \ra C$ in $\D$ with $f'\c t_{A,B} = f$. We write $a\t b$ for the element $t_{A,B}(a,b)\in\under{A\t B}$.
\item Tensor products 
exist in any variety $\D$. Let us indicate how to construct them 
in the case where $\D$ is a variety of ordered algebras. Given $A,B\in \D$, 
form the free algebra 
$\Psi(\under{A}\times\under{B})$ in $\D$ generated by the set
$\under{A}\times\under{B}$. (For simplicity, we assume that 
$\under{A}\times\under{B}$ is a subset of $\Psi(\under{A}\times\under{B})$.) 
Form the smallest stable preorder $\preceq$ on
$\Psi(\under{A}\times\under{B})$ containing all inequations of the form
\begin{align*}
(a,\gamma(b_1,\ldots,b_n)) &\preceq \gamma((a,b_1),\ldots,(a,b_n))\\
\gamma((a,b_1),\ldots,(a,b_n)) &\preceq (a,\gamma(b_1,\ldots,b_n))\\
(\gamma(a_1,\ldots,a_n),b) &\preceq (\gamma((a_1,b),\ldots,(a_n,b))\\ 
(\gamma((a_1,b),\ldots,(a_n,b))&\preceq  (\gamma(a_1,\ldots,a_n),b) 
\end{align*}
where $\gamma\in \Gamma$ is an $n$-ary operation symbol, $a,a_1,\ldots,a_n\in 
A$ and $b,b_1,\ldots,b_n\in B$. Then the tensor product of $A$ and $B$ is 
given by \[A\t B := 
\Psi(\under{A}\times\under{B})/\mathord{\preceq}\] and the universal bimorphism \[t_{A,B} := 
(\,\under{A}\times\under{B}\monoto \Psi(\under{A}\times\under{B}) \xra{\pi} 
A\t B =
\Psi(\under{A}\times\under{B})/\mathord{\preceq}\,),\]
the composite of the inclusion map and the projection. In particular, $A\t B$ 
is generated by the elements $a\t b$ with $a\in\under{A}$ and $b\in\under{B}$.
The construction of 
$A\t B$ for unordered algebras is analogous: just replace inequations and 
stable 
preorders by equations and 
congruences.
\item For any $A,B,C\in \D$ there is a natural bijective correspondence between (i) 
morphisms from $A\t B$ to $C$, (ii) bimorphisms from $A,B$ to $C$, and (iii) 
morphisms 
from $A$ to $[B,C]$. Indeed, the correspondence of (i) and (ii) follows from 
the universal property of the tensor product. Further, any bimorphism $f: 
\under{A}\times \under{B}\to \under{C}$ defines a morphism
\[ \lambda f: A\to [B,C],\quad (\lambda f)(a)(b) = f(a,b), \]
and the map $f\mapsto \lambda f$ gives the bijective correpondence between 
(ii) 
and 
(iii).
\item Up to isomorphism, $\t$ is associative, commutative and has unit $\one_\D$.
More precisely, for any three objects 
$A,B,C\in \D$ there are natural isomorphisms $\alpha_{A,B,C}$, $\sigma_{A,B}$,
$\rho_A$ and $\lambda_A$ making the 
following squares commute:
\[ 
\xymatrix{
\under{(A\t B)\t C} \ar[r]^{\alpha_{A,B,C}} & \under{A\t (B\t C)}\\
(\under{A}\times \under{B})\times \under{C} \ar[u]^{t_{A\t B, C} 
\c(t_{A,B}\times C)} \ar[r]_{\alpha'_{A,B,C}} & \under{A} 
\times (\under{B}\times \under{C}) \ar[u]_{t_{A,B\t C} \c(A\times t_{B,C})}
}
\]
\[ 
\xymatrix{
\under{A\t B} \ar[r]^{\sigma_{A,B}} & \under{B\t A}\\
\under{A}\times \under{B} \ar[u]^{t_{A,B}} \ar[r]_{\sigma'_{A,B}} & 
\under{B}\times 
\under{A} \ar[u]_{t_{B,A}}
}\quad
\xymatrix{
\under{A}\t \under{\one_\D} \ar[r]^{\rho_A} & \under{A}\\
\under{A}\times \under{\one_\D} \ar[u]^{t_{A,\one_\D}} \ar[ur]_{\pi_A} &
}
\quad
\xymatrix{
\under{\one_\D\t A} \ar[r]^{\lambda_A} & \under{A}\\
 \under{\one_\D}\times \under{A} \ar[ur]_{\pi_A'} \ar[u]^{t_{\one_\D,A}} &
}
\]
where $\alpha'$ and $\sigma'$ are the canonical bijections, and $\pi_A$ and 
$\pi_A'$ are the projection maps. We shall often omit indices and write
$t$ for $t_{A,B}$, $\alpha$ for $\alpha_{A,B,C}$, etc.
\item A \emph{$\D$-monoid} $(M,\iota,\mu)$ is 
an object $M\in \D$ equipped with two morphisms $\iota: \one_\D\to 
M$ and $\mu: M\t M\to M$ such that the following diagrams commute:
\[
  \vcenter{
\xymatrix@C-2.5em{
  & (M\t M) \t M \ar[rr]^{\alpha} \ar[ld]_-{\mu \t M} & & M \t (M \t M)
  \ar[rd]^{M \t \mu} \\
  M \t M \ar[rrd]_{\mu} & & & & M \t M  \ar[lld]^{\mu} \\
  & & M 
}
}\quad
\vcenter{
\xymatrix{
  M \t \one_\D \ar[rd]^{\lambda} \ar[d]_{M \t \iota} \\
  M \t M \ar[r]^{\mu} & M\\
  \one_\D \t M \ar[ru]_{\rho} \ar[u]^{\iota  \t M}
}
}
\]
A \emph{morphism} between $\D$-monoids $(M,\iota_M,\mu_M)$ and $(N,\iota_N,\mu_N)$ is 
a morphism 
$h: M\to N$ in $\D$ such that the following square commutes:
\[
\xymatrix@C+1em{
M\t M \ar[r]^>>>>>>>{\mu_M} \ar[d]_{h\t h} & M \ar[d]^h & \ar[dl]^{\iota_N} \ar[l]_{\iota_M}\one_\D \\
N\t N \ar[r]_>>>>>>>{\mu_N} & N &
}
\]
Due to $\t$ representing bimorphisms, the notion of $\D$-monoids and their morphisms given here is equivalent to the set-theoretic one of Definition \ref{def:dmonoid}. 
\end{enumerate}
\end{appsec}

\section{Categorical Toolkit}
We assume familiarity with basic concepts from category theory, like 
categories, 
functors,  
natural transformations, and (co-)limits (see e.g. \cite{macl}). However, we 
recall here some definitions and facts concerning adjunctions.

\begin{appsec}\label{app:adjunctions}\textbf{Adjunctions.}
Let $U: \A\to \X$ be a functor between categories $\A$ and $\X$. Suppose that there 
exists,  for each 
$X\in\X$,  an object 
$FX\in\A$ and a morphism $\eta_X: X\to U(FX)$ in $\X$ with the following universal 
property: for any morphism $f: A\to UB$ in $\X$ with $B\in\B$, there is a 
unique $\ol f: FX\to B$ in $\B$ (called the \emph{adjoint transpose} of $f$) 
with $U(\ol f)\c \eta_X = f$. In this case the 
object map $X\mapsto FX$ extends uniquely to a functor $F: \X\to\A$ such that 
$\eta: \Id_\X \to UF$ becomes a natural transformation, and we say that the
functors $U$ and $F$ form an \emph{adjunction} (commonly denoted by 
$F\dashv U: 
\A\to\X$). The functor $U$ is the \emph{right adjoint}, $F$ the 
\emph{left adjoint}, and the natural transformation $\eta$ the \emph{unit} of 
the adjunction. $\eta$ 
induces 
another natural transformation $\epsilon: FU \to \Id_\A$ with components 
$\epsilon_A := \ol{\id_{UA}}$, called the \emph{counit} of the adjunction. The 
universal property gives rise to an isomorphism $\A(FX,A)\cong \X(X,UA)$ 
natural in $X\in \X$ and $A\in A$. An important fact about adjunctions is that 
right adjoints preserve limits, and dually left adjoints preserve colimits 
(and 
thus, in particular, epimorphisms).

A typical source of adjunctions are free constructions in algebra. For any 
variety 
$\D$ of algebras or ordered algebras, the forgetful functor $\under{\dash}: 
\D\to\Set$ (mapping an algebra to its underlying set) has the left adjoint 
$\Psi: \Set\to\D$ that maps a set $X$ to the free algebra $\Psi X$ in $\D$ 
generated by $X$. The unit $\eta_X: X\to \under{\Psi X}$ is the inclusion of 
generators. The freeness of $\Psi X$ amounts exactly to the universal property 
of 
$\eta_X$.
\end{appsec}

\begin{appsec}\label{app:compadjunctions}\textbf{Composition of adjunctions.} The \emph{composite} of two 
adjunctions \[(F\dashv U: \A\to \B,\eta,\epsilon)\quad\text{and}\quad(G\dashv 
V: \B\to \X, 
\eta', \epsilon')\] is the adjunction $FG\dashv VU: \A\to \X$ with unit 
$X\xra{\eta'_X} VGX \xra{V\eta_{GX}} VUFGX$ ($X\in \X$) and counit $FGVUA 
\xra{F\epsilon'_{UA}} FUA \xra{\epsilon_A} A$ ($A\in \A$).
\end{appsec}

\begin{appsec}\label{app:yoneda}\textbf{Yoneda lemma (weak form).} Let $\A$ be a category and $A\in \A$. The \emph{hom-functor} $\A(A,\dash): \A\to\Set$ maps an object $B\in \A$ to the set $\A(A,B)$ of morphisms from $A$ to $B$, and a morphism $f: B\to B'$ to the function
\[ \A(A,f): \A(A,B)\to \A(A,B'),\quad g\mapsto f\c g. \]
The hom-functor determines objects of $\A$ up to isomorphism: any natural isomorphism $\theta: 
\A(A,\dash)\cong \A(A',\dash)$ with $A,A'\in  \A$ yields an isomorphism 
$\theta_A(\id_A): A'\cong A$.
\end{appsec}

\section{Monoidal Adjunctions}

Here we present some well-known facts about monoidal functors and adjunctions. 
Since these facts appear scattered throughout the literature or are folklore 
in category theory, we sketch the proofs for some statements for the 
convenience of the reader. In the following, let $\A$, $\B$, $\C$, $\D$ be 
commutative varieties of (ordered) algebras; we remark that all concepts 
treated in this section can be introduced in  a more general form for monoidal 
categories, see e.g. \cite{macl}.

Recall from Definition \ref{def:monoidalfunc} the notion of a \emph{monoidal functor} and a \emph{monoidal natural transformation}. Recall also from \ref{app:tensorproducts} that $\D$-monoids $(M,1,\bullet)$ (see Definition \ref{def:dmonoid}) can be represented as triples $(M,\iota,\mu)$.

\begin{lemma}\label{lem:monlifting2}
Any monoidal functor $(G,\theta): \A\to \B$ lifts to a functor $\ol G: 
\Mon{\A}\to 
\Mon{\B}$ such that the following diagram commutes, where $U_\A$ and $U_\B$ 
are the forgetful functors:
\[
\xymatrix{
\Mon{\A}  \ar[d]_{U_\A} \ar[r]^{\ol G} & \Mon{\B} \ar[d]^{U_\B} \\
\A \ar[r]_G & \B
} 
\]
Explicitly, $\ol G$ maps an $\A$-monoid $(M,\iota,\mu)$ to the $\B$-monoid 
\[
  (GM,\ \one_\B\xra{\theta} G\one_\A \xra{G\iota} GM,\ GM\t GM \xra{\theta} 
  G(M\t M)\xra{G\mu}
  GM),
\]
and an $\A$-monoid morphism $f$ to $Gf$. Moreover, any monoidal natural 
transformation $\phi: (G,\theta)\to (H,\sigma)$ yields a natural 
transformation $\ol\phi: \ol G\to \ol H$ with components
\[ \ol\phi_{(M,\iota,\mu)} = \phi_M: \ol G(M,\iota,\mu) \to \ol H(M,\iota,\mu). \]
\end{lemma}
\begin{proof}
  \begin{enumerate}
    \item It is straightforward to show that $(GM, G\iota \c \theta, G\mu \c
      \theta_{M, M})$ is a $\B$-monoid. For example, associativity is 
      established by the commutative diagram below. Part $(1)$ commutes since 
      $G$ is monoidal; $(2)$ commutes since $(M,
      \iota, \mu)$ is a monoid; $(3)$ and $(4)$ commute because $\theta_{A, 
      B}\colon
      GA \t GB \to G(A \t B)$ is natural in $A$ and $B$.
      \[
        \xymatrix{
        & (GM\t GM)\t GM \ar[rr]^{\alpha_{GM, GM, GM}} \ar@{}[rrdd]|{(1)}
        \ar[d]_{\theta_{M, M}\t GM} & & GM \t (GM\t GM) \ar[d]^{GM\t \theta_{M,
            M}}\\
        & G(M\t M) \t GM \ar[d]_{\theta_{M\t M, M}} \ar@/_2pc/[ld]_-{G\mu \t GM} &
        & GM \t G(M\t M) \ar[d]^{\theta_{M, M\t M}} \ar@/^2pc/[rd]^-{GM \t G\mu}
        \\1
        GM \t GM \ar@{}[r]|-{(3)} \ar@/_2pc/[rd]_{\theta_{M, M}} & G((M\t M)\t
        M) \ar[d]_{G(\mu \t M)} \ar[rr]_{G\alpha_{M, M, M}}& \ar@{}[dd]|{(2)} &
        G(M\t (M\t M)) \ar[d]^{G(M \t \mu)} \ar@{}[r]|-{(4)} & GM \t GM
        \ar@/^2pc/[ld]^{\theta_{M, M}}\\
        & G(M \t M) \ar[rd]_{G\mu} & & G(M\t M) \ar[ld]^{G\mu} \\
        & & GM
        }
      \] 
      The unit laws
      follow in a similar way.

      Assume that $f\colon (M,\iota_M, \mu_M) \to (N, \iota_N, \mu_N)$ is an $\A$-monoid morphism.
      Then the following diagram commutes:
      \[
        \xymatrix{
          GM \t GM \ar[r]^{\theta_{M, M}} \ar[d]_{Gf \t Gf} & G(M\t M) \ar[r]^-{G\mu_M}
          \ar[d]_{G(f\t f)} & GM \ar[d]^{Gf} \\
          GN \t GN \ar[r]_{\theta_{N, N}} & G(N\t N) \ar[r]_-{G\mu_N} & GN
        }
      \]
      where the right square commute because $f$ is a $\D$-monoid morphism and 
      the left
      square commutes as $\theta_{M, M}$ is natural. Together with the 
      corresponding diagram for the preservation of the unit, this shows that 
      $Gf: \ol G(M,\iota_M,\mu_M)\to \ol G(N,\iota_M,\mu_N)$ is a $\B$-monoid 
      morphism.
    \item To show that every monoidal natural transformation $\varphi$ 
      lifts to a natural transformation from $\ol G$ to $\ol H$, it suffices to
      show that $\varphi_M$ is a $\B$-monoid morphism for every $\A$-monoid 
      $(M, \iota,
      \mu)$. The preservation of the multiplication follows from the following diagram:
      \[
        \xymatrix@C+2em{
          GM \t GM \ar[r]^{\varphi_M \t \varphi_M} \ar[d]_{\theta_{M, M}}
          & HM \t HM \ar[d]^{\sigma_{M, M}} \\
          G(M \t M) \ar[r]^{\varphi_{M \t M}} \ar[d]_{G\mu} & H(M \t M) \ar[d]^{H\mu}
          \\
          GM \ar[r]_{\varphi_M} & HM 
        }
      \]
      where the upper square uses that $\phi$ is a monoidal natural
      transformation and the lower square is the naturality of~$\varphi$. Similarly for the preservation of the unit.
  \end{enumerate}
\end{proof}

\begin{lemma}\label{lem:moncompose}
For any two monoidal functors $(G,\theta)\colon \A\to \B$ and $(H,\sigma)\colon
\B\to\C$ the composite $HG$ becomes a monoidal functor via
\begin{align*}
  (H\theta\c\sigma)_{A, B} & =
  H(GA)\t H(GB) \xra{\sigma_{A, B}} H(GA\t GB) \xra{H\theta_{GA, GB}} HG(A\t B)
  \\
  G\theta \c \sigma & = \one\C  \xra{\sigma} H\one_\B \xra{G\theta} HG\one_\A.
\end{align*}
\end{lemma}
\begin{proof}
  The naturality of $(H\theta \c \sigma)_{A, B}$ follows from the naturality of
  $\sigma_{A, B}$ and~$H\theta_{GA, GB}$. It remains to verify the diagrams in Definition \ref{def:monoidalfunc}. For example, for the left diagram we get
  \[
    \xymatrix@+4em@ru{
      & (HGA \t HGB) \t HGC \ar[r]^{\alpha_{HGA, HGB, HGC}}
      \ar[d]_{\sigma_{GA, GB} \t HGC} \ar@{}[rdd]|{(1)}
      & HGA \t (HGB \t HGC) \ar[d]^{HGA \t \sigma_{HGB, HGC}} \\
       & H(GA \t GB) \t HGC \ar[d]^{\sigma_{GA\t GB, GC}}
      \ar[ld]_{H\theta_{A, B} \t HGC}
      & HGA \t H(GB \t GC) \ar[d]^{\sigma_{GA, GB \t GC}}
      \ar[rd]^{HA \t H\theta_{B, C}} \\
      HG(A \t B) \t HGC \ar[rd]_{\theta_{GA\t GB, GC}}
      \ar@{}[r]|{(3)}
      & H((GA \t GB) \t GC) \ar[r]_{H\alpha_{GA, GB, GC}} \ar[d]_{H(\theta_{A,
          B} \t GC)}\ar@{}[rdd]|{(2)}
      & H(GA \t (GB \t GC)) \ar[d]^{H(GA \t \theta_{B, C}}  
      \ar@{}[r]|{(4)}
      & HGA \t HG(B \t C) \ar[ld]^{\theta_{GA, G(B\t C)}} 
       \\
      & H(G(A \t B) \t GC) \ar[d]_{H\theta_{A \t B, C}}
      & H(GA \t G(B \t C)) \ar[d]_{H\theta_{A, B \t C}}
      &  
      \\
      & HG((A \t B) \t C) \ar[r]_{HG\alpha_{A, B, C}} & HG(A \t (B \t C))
    }
  \]
  In $(1)$ we use that $H$ is monoidal, in $(2)$ that $G$ is monoidal, in (3) 
  the naturality of $\theta$, and in (4) the naturality of $\sigma$. Similarly 
  for the other two diagrams in Definition \ref{def:monoidalfunc}.
\end{proof}



%

\begin{expl}\label{ex:upsimonoidal}
\begin{enumerate}[(1)]
\item
For any commutative variety $\D$ of algebras or ordered algebras, the forgetful functor $\under{\dash}: \D\to \Set$ is a monoidal 
functor w.r.t. the maps
\[ i: 1\monoto \under{\one_\D} = \under{\Psi 1}\quad\text{and}\quad t_{A,B}: \under{A}\times\under{B}\to \under{A\t B}\]
where $i$ is the inclusion of the generator.
Indeed, by Definition \ref{def:monoidalfunc} we need to verify that the following diagrams commute for all $A,B\in \D$:
\begin{center}
\begin{tabular}[x]{cc}
\parbox{0.6\textwidth}{
$
\xymatrix@C-1em@R-1.5em{
(\under{A}\times \under{B})\times \under{C} \ar[r]^{\alpha'} \ar[d]_{t \times \under{C}} & \under{A} \times 
(\under{B}\times
\under{C}) \ar[d]^{\under{A}\times t}\\
\under{A\t B} \times \under{C} \ar[d]_{t} & \under{A} \times \under{B\t C} 
\ar[d]^t \\
\under{(A\t B)\t C} \ar[r]_{\alpha}& \under{A\t (B\t C)}
}
$
}
&
\parbox{0.4\textwidth}{
$
\xymatrix@R-1.5em{
\under{A}\times 1 \ar[r]^<<<<<{\under{A}\times i} \ar[d]_{\rho'} & \under{A}\times \under{\one_\D} 
\ar[d]^t\\
\under{A} & \under{A\t \one_\D} \ar[l]_<<<<<<<{\rho}
}
$\\
$\xymatrix@R-1.5em{
1\times\under{A} \ar[r]^<<<<<{i\times \under{A}} \ar[d]_{\lambda'} & \under{\one_\D}\times \under{A} 
\ar[d]^t\\
\under{A} & \under{\one_\D\t A} \ar[l]_<<<<<<<{\lambda}
}
$
}
\end{tabular}
\end{center}

But this follows immediately from the definitions of $t$, $\alpha$, $\rho$, $\lambda$, see \ref{app:tensorproducts}. For example, both legs of the left hand diagram map an element $((a,b),c)\in (\under{A}\times\under{B})\times \under{C}$ to $a\t (b\t c)$. Similarly for the other two diagrams.
\item The left adjoint $\Psi: \Set\to\D$ to $\under{\dash}: \D\to \Set$ is also monoidal. Indeed, observe that for any two sets $X$ and 
$Y$ we have the following bijections (natural in $A\in\D$), cf. \ref{app:tensorproducts}(3):
\begin{align*}
 \D(\Psi(X\times Y), A) &\cong \Set(X\times Y, \under{A})\\
 &\cong \Set(X,\under{A}^Y) \\
 &\cong \Set(X, |[\Psi Y, A]|)\\
 &\cong \D(\Psi X, [\Psi Y, A]) \\
 &\cong \D(\Psi X\t \Psi Y, A).
\end{align*}
The Yoneda lemma, see \ref{app:yoneda}, gives a natural isomorphism $\Psi(X\times Y)\cong \Psi X\t \Psi Y$, mapping a pair $(x,y)\in X\times Y$ to 
$x\t y\in \under{\Psi X\t \Psi Y}$. Its inverse $\theta_{X,Y}:  \Psi X\t \Psi 
Y \cong \Psi(X\times Y)$ together with the morphism $\theta = \id: 
\one_\D\to \Psi 1$, makes $\Psi$ a monoidal functor, i.e. the following diagrams commute for all sets $X,Y,Z$: 
\begin{center}
\begin{tabular}[x]{cc}
\parbox{0.6\textwidth}{
$
\xymatrix@C-1em@R-1.5em{
(\Psi{X}\t \Psi{Y})\t \Psi{Z} \ar[r]^{\alpha} \ar[d]_{\theta \times \Psi{Z}} & \Psi{X} \t 
(\Psi{Y}\t
\Psi{Z}) \ar[d]^{\Psi{X}\t \theta}\\
\Psi(X\times Y) \t \Psi{Z} \ar[d]_{\theta} & \Psi{X} \t \Psi(Y\times Z) 
\ar[d]^\theta \\
\Psi((X\times Y)\times Z) \ar[r]_{\Psi\alpha'}& \Psi(X\times (Y\times Z))
}
$
}
&
\parbox{0.4\textwidth}{
$
\xymatrix@R-1.5em{
\Psi{X}\t \one_\D \ar[r]^<<<<<{\Psi{X}\t \id} \ar[d]_{\rho} & \Psi{X}\t \Psi 1
\ar[d]^\theta\\
\Psi{X} & \Psi(X\times 1) \ar[l]_<<<<<<<{\Psi\rho'}
}
$\\
$
\xymatrix@R-1.5em{
\one_\D\t \Psi X \ar[r]^<<<<<{\id\t \Psi X} \ar[d]_{\lambda} & \Psi 1\t \Psi X
\ar[d]^\theta\\
\Psi{X} & \Psi(1\times X) \ar[l]_<<<<<<<{\Psi\lambda'}
}
$
}
\end{tabular}
\end{center}
This follows directly from the definitions of $\theta$, $\alpha$, $\lambda$, 
$\rho$. For example, both legs of the left-hand diagram map an element $(x\t 
y)\t z)\in \under{(\Psi X\t \Psi Y)\t \Psi Z}$ (with $x\in X$, $y\in Y$, $z\in 
Z$) to $(x,(y,z)) \in X\times (Y\times Z)\seq \Psi(X\times(Y\times Z))$. Since 
the elements $(x\t y)\t z$ generate the algebra ${(\Psi X\t \Psi Y)\t \Psi 
Z}$, see \ref{app:tensorproducts}, this shows that the diagram commutes. 
Similarly for the other two diagrams.
\item The adjunction $\Psi\dashv \under{\dash}:\D\to\Set$ is monoidal (see 
Definition \ref{def:monoidaladj}). To see this, we need to show that the unit 
$\eta: \Id_\Set\to \under{\Psi}$ and the counit $\epsilon: \Psi\c 
\under{\dash} \to  \Id_\D $ are monoidal natural transformations. We only 
prove 
that $\epsilon$ is monoidal, since the proof for $\eta$ is similar. Note first 
that by $\Psi\c \under{\dash}: \D\to\D$ is meant the composite monoidal 
functor 
in the sense of Lemma \ref{lem:moncomp}. Thus the associated morphisms are
\[ \Psi\under{A}\t \Psi\under{B} \xra{\theta_{\under{A},\under{B}}} 
\Psi(\under{A}\times \under{B}) \xra{\Psi t_{A,B}} \Psi\under{A\t B} \]
and 
\[ \one_\D \xra{\id=\theta_1} \Psi 1 \xra{\Psi i} \Psi\under{\one_\D}. \]
To show that $\epsilon$ is a monoidal natural transformation, we need to 
verify that the following 
two 
diagrams commute, cf. Definition \ref{def:monoidalfunc}.
\[
\xymatrix{
\Psi{\under{A}}\t\Psi\under{B} \ar[r]^>>>>{\epsilon_A\t \epsilon_B} 
\ar[d]_{\Psi{\under{A},\under{B}}} & A\t B \ar@{=}[dd]\\
\Psi(\under{A}\times\under{B}) \ar[d]_{\theta t_{A,B}} & \\
\Psi\under{A\t B} \ar[r]_{\epsilon_{A\t B}} & A\t B
}\qquad\qquad
\xymatrix{
& \one_\D \ar[dl]_{\id} \ar@{=}[ddr] &\\
\Psi 1 \ar[d]_{\Psi i} & &\\
\Psi\under{\one_\D} \ar[rr]_{\epsilon_{\one_\D}} && \one_\D
}
\]
But this follows immediately from the definitions of $\epsilon$, $\theta$ and 
$i$. Both legs of the left diagram map an element $a\t b\in 
\Psi\under{A}\t\Psi\under{B}$ (with $a\in \under{A}$ and $b\in\under{B}$) to 
$a\t b$. Since $\Psi\under{A}\t\Psi\under{B}$ is generated by these elements, 
this shows that the left diagram commutes. In the right diagram, both legs map 
the generator of $\one_\D$ to itself, and thus both legs are the identity 
morphism.
\end{enumerate}
\end{expl}
\section{Proofs}

\begin{notation}
Throughout this section, we assume the categorical setting introduced at the 
beginning of Section 3, i.e. we have a concrete monoidal adjunction $F\dashv 
U: \Mod{\S}\to\D$ with unit $\eta:\Id \to UF$.  We denote the forgetful 
functors by $\under{\dash}_\S: \Mod{\S}\to\Set$ and $\under{\dash}_\D: \D 
\to\Set$, but usually drop the indices. For a morphism $f$ in $\D$ or 
$\Mod{\S}$ we write
$f$ for $\under{f}$, i.e. there is no notational distinction between a 
morphism and its underlying function. This is legimitate because the forgetful 
functors are faithful, i.e. every function $\under{A}\to\under{B}$ 
carries at most one morphism from $A$ to $B$. The morphisms 
witnessing that $U$ and $F$ are monoidal 
functors are denoted by
\[ \theta^U: \one_\D \to U\S,\quad \theta^U_{A,B}: UA\t UB\to U(A\t B), \]
and
\[  \theta^F: \S\to F(\one_\D),\quad \theta^F_{X,Y}: FX\t FY\to F(X\t Y).\]
Here $\S$ is viewed as the free $\S$-module on one generator. Finally, we fix two $\D$-monoids $(M,1,\bullet)$ and $(N,1,\bullet)$. We usually write $xy$ for $x\bullet y$.
\end{notation}

\begin{rem}\label{rem:plusmorph}
Our assumption that $F\dashv U: \Mod{\S}\to \D$ is a concrete monoidal adjunction has the following consequences:
\begin{enumerate}[(1)]
\item An $\S$-module $A$ and the corresponding algebra $UA$ in $\D$ have the 
same underlying set, i.e. $\under{A}_\S = \under{UA}_\D$. Likewise, an 
$\S$-module morphism $f:A\to B$ and the 
corresponding morphism 
$Uf: UA\to UB$ in $\D$ have the same underlying function.
\item Every bimorphism from $A,B$ to $C$ in $\Mod{\S}$ is also a bimorphism from $UA,UB$ to $UC$ in $\D$. To see this, suppose that $f: \under{A}_\S\t\under{B}_\S\to \under{C}_\S$ is a bimorphism of $\Mod{\S}$, and let $f': A\t B\to C$ be the corresponding $\S$-linear map. Consider the diagram below:
\[
\xymatrix{
\under{UA\t UB}_\D \ar[r]^{\theta^U} &  \under{U(A\t B)}_\D \ar[r]^<<<<<<{Uf'} \ar@{=}[d] & \under{UC}_\D \ar@{=}[d]\\
& \under{A\t B}_\S \ar[r]^{f'}   & \under{C}_\S\\
\under{UA}_\D\times\under{UB}_\D \ar[uu]^{t_\D} \ar@{=}[r] & \under{A}_\S\times\under{B}_\S \ar[u]^{t_\S}  \ar[ur]_{f} &
}
\]
Here $t_\D$ and $t_\S$ denote the universal bimorphisms.
The left part commutes because the monoidal functor $\under{\dash}_\S$ is the 
composite of the monoidal functors $U$ and $\under{\dash}_\D$ (as $U$ is part 
of a concrete monoidal adjunction); cf. Lemma \ref{lem:moncompose} and Example 
\ref{ex:monoidalfunc}(1). The upper right square commutes because $f'$ and 
$Uf'$ have the same underlying function, see (1). The triangle is the 
definition of $f'$. It follows that $f: \under{UA}_\D\times 
\under{UB}_\D\to\under{UC}_\D$ is a bimorphism from $UA,UB$ to $UC$ in $\D$, 
being the composite of the $\D$-bimorphism $t_\D$ with the $\D$-morphism 
$Uf'\c \theta^U$.
\item As a consequence of (2), for any $\S$-algebra $A$ the multiplication 
$\under{A}\times\under{A}\xra{\o}\under{A}$ is a bimorphism from $UA, UA$ to 
$UA$. Morever, the sum $\under{A}\times \under{A}\xra{+} \under{A}$ 
carries a morphism $+: A\times A\to A$ in $\Mod{\S}$ because the latter is a 
commutative variety. Applying $U$, we obtain a morphism $+: UA\times UA\to UA$ 
in $\D$. Analogously for the scalar product $\lambda\o\dash: 
A\to A$ $(\lambda\in S)$.
\end{enumerate}
\end{rem}

\begin{rem}
For any two morphisms $p: M\to \S$ and $q: N\to \S$ in $\D$ and elements $m\in 
\under{M}$ and $n\in\under{N}$, we have 
\begin{equation}\label{eq:tensorformula}(p\c (\dash\bullet m))\t (q\c 
(\dash\bullet n)) = (p\t q)\c (\dash\bullet (m\t n))
\end{equation}
where $\dash\bullet m: M\to M$, $\dash\bullet n: N\to N$ and $\dash\bullet 
(m\t n): M\t N\to M\t N$ are morphisms of $\D$ because $M$, $N$ and $M\t N$ 
are $\D$-monoids (i.e. their multiplication is a $\D$-bimorphism). Indeed, for 
any $m'\in\under{M}$ and $n'\in \under{N}$ we have
\begin{align*}
[p\c (\dash\bullet m))\t (q\c 
(\dash\bullet n))](m'\t n') &= p(m'\bullet m) \t q(n'\bullet n)\\
&= p\t q((m'\t n')\bullet (m\t n))\\
&= [(p\t q)\c (\dash \bullet (m\t n))](m'\t n').   
\end{align*}
In the first step we use the definition of the tensor product of two 
morphisms, see Remark \ref{rem:tensorproducts}. In the second step, we use the 
definition of the multiplication in $M\t N$, see Remark 
\ref{rem:tensormonoid}. The last step is obvious. Since the elements $m'\t n'$ 
generate $M\t N$, see \ref{app:tensorproducts}, this proves 
\eqref{eq:tensorformula}.
\end{rem}

\begin{proof}[Lemma \ref{lem:mastn}]
We prove the case where $\D$ is a variety of ordered algebras, the unordered 
case being analogous. Consider the preorder $\preceq$ on $M\t N$ defined 
$x\preceq y$ iff $\pi(x)\leq \pi(y)$. By \ref{app:congruences} we only need to 
show 
that $\preceq$ is a stable preorder of the $\D$-monoid $M\t N$ (cf. Remark 
\ref{rem:tensormonoid});  then, since $\pi$ is surjective,
 \[\pi(x)\bullet 
\pi(y) := \pi(x\bullet y)\quad (x,y\in \under{M\t N})\]
 gives a well-defined 
$\D$-monoid structure on $M\ast N$ making $\pi$ a $\D$-monoid morphism. In 
particular for $x=m\t n$ and $y=m'\t n'$, the multiplication is given by 
\begin{align*}
(m\ast n)\bullet (m'\ast n') &= \pi(m\t n)\bullet \pi(m'\t n')\\
& = \pi((m\t n)\bullet (m'\t n'))\\
& = \pi((mm')\t (nn'))\\
&=  (mm')\ast (nn')
\end{align*}
Thus let us show that $\preceq$ is indeed stable. Clearly $x\leq y$ in $M\t N$ 
implies $x\preceq y$ 
because $\pi$ is monotone. Also, $\preceq$ is stable w.r.t. all operations 
of 
$\D$, since $\pi$ is morphism of $\D$. It remains to show that $x\preceq y$ 
implies $x\bullet z\preceq y\bullet z$ and $z\bullet 
x\preceq z\bullet y$ (equivalently, $\pi(x)\leq 
\pi(y)$ 
implies $\pi(x\bullet z)\leq \pi(y\bullet z)$ and $\pi(z\bullet x)\leq 
\pi(z\bullet y)$) for all $x,y,z\in \under{M\t N}$.  We may assume that 
$z=m\t n$ for some $m\in\under{M}$ and $n\in \under{N}$; since $M\t N$ is 
generated by 
these elements, see \ref{app:tensorproducts}, and $\bullet$ is a $\D$-bimorphism, this implies the statement 
for all $z$. So suppose that $\pi(x)\leq \pi(y)$. By \eqref{eq:s1s2} this 
implies that
\[ \sigma\c (p\t q) (x) \leq \sigma\c (p\t q)(y)\quad\text{for all $p: M\to \S$ and $q: N\to \S$ in 
$\D$}. \]
In particular  we get, for all $p$ and $q$,
\[ \sigma\c (\,(p\c (\dash\bullet m))\t (q\c (\dash\bullet n))\,)(x)\leq \sigma\circ (\,(p\circ 
(\dash\bullet m))\t (q\c (\dash\bullet n))\,)(y) \]
using that $\dash\bullet m: M\to M$ and $\dash\bullet n: N\to N$ are morphisms 
of $\D$ since the multiplication of $M$ resp. $N$ is a bimorphism. 
Equivalently, by \eqref{eq:tensorformula},
\[ \sigma\c (p\t q)\c (\dash \bullet (m\t n))(x)\leq \sigma\c (p\t q)\c (\dash 
\bullet (m\t n))(y)\]
 Thus, since $z=m\t n$,
\[ \sigma\c (p\t q)(x \bullet z)\leq \sigma\c (p\t q) (y 
\bullet z).\]
for all $p$ and $q$.
 By the definition of $\pi$, this means precisely that $\pi(x\bullet z)\leq \pi(y\bullet z)$. The proof of  $\pi(z\bullet x)\leq \pi(z\bullet y)$ is symmetric.
\end{proof}

\begin{lemma}\label{lem:afg}
For any $\S$-algebra $A$ and any two $\D$-monoid morphisms $f: M \to \ol U A$ 
and 
$g: 
N\to \ol U A$, 
the product  
$M\times UA \times N$ in $\D$ carries a $\D$-monoid structure with unit $(1,0,1)$ and multiplication
\[ (m,a,n)(m',a',n') := (mm',\,f(m)\o a'+a\o g(n'),\,nn').  \]
Here the $0$ and the sum and product in the second component are taken in the $\S$-algebra $A$. Denoting this $\D$-monoid by $A^{f,g}$, the product projections $\pi_M: A^{f,g}\to M$ and $\pi_N: A^{f,g}\to N$ are $\D$-monoid morphisms.
\end{lemma}

\begin{proof}
That $(1,0,1)$ is the unit is clear since $f(1)=1$ and $g(1)=1$. For associativity, we 
compute
\begin{align*}
&[(m,a,n)(m',a',n')](m'',a'',n'')\\
&= (mm',\,f(m)\o a'+a\o g(n'),\,nn')(m'',a'',n'') \\
&= ((mm')m'',\,f(mm')\o a'' + [f(m)\o a'+a\o g(n')]\o g(n''),\,(nn')n'')\\
&= (m(m'm''),\,f(m)\o [f(m')\o a''+a'\o g(n'')] + a\o g(n'n''),\,n(n'n''))\\
&= (m,a,n)(m'm'',\,f(m')\o a''+a'\o g(n''),\,n'n'')\\
&= (m,a,n)[(m',a',n')(m'',a'',n'')]
\end{align*}
It remains to verify that the multiplication of $A^{f,g}$ is a 
$\D$-bimorphism. For simplicity, let us just prove that for any binary operation $\gamma$ in the signature of $\D$, the multiplication preserves $\gamma$ in the right component. Indeed, we have
\begin{align*}
& (m,a,n)[\gamma((m',a',n'),(m'',a'',n''))] & \\
&= (m,a,n)(\gamma(m',m''),\,\gamma(a',a''),\,\gamma(n',n'')) & (1)\\
&= (m\gamma(m',m''),\,f(m)\o\gamma(a',a'')+a\o g(\gamma(n',n'')),\,n\gamma(n',n'')) &(2)\\
&= (\gamma(mm',mm''),\,\gamma(f(m)\o a',f(m)\o a'')+ 
a\o \gamma(g(n'),g(n'')),\,\gamma(nn',nn''))& (3)\\
&= (\gamma(mm',mm''),\,\gamma(f(m)\o a',f(m)\o a'')+ 
\gamma(a\o g(n'),a\o g(n'')),\,\gamma(nn',nn'')) & (4)\\
&= (\gamma(mm',mm''),\,\gamma(f(m)\o a'+a\o g(n'),f(m)\o a''+a\o g(n'')),\,\gamma(nn',nn'')) &(5)\\
&= \gamma(\,(mm',f(m)\o a'+a\o g(n'), nn'),\,(mm'', 
f(m)\o a''+a\o g(n''), nn'')\,) & (6)\\
&= \gamma((m,a,n)(m',a',n'), (m,a,n)(m'',a'',n'')) & (7)
\end{align*}
Explanation of the individual steps:
(1) Definition of the operation $\gamma$ in the product $M\times UA\times N$ 
in $\D$.
(2) Definition of the multiplication in $A^{f,g}$.
(3) In the first and third component we use that the multiplication of $M$ and 
$N$ is a $\D$-bimorphism; in the second component we use that the 
multiplication of $A$ is a $\D$-bimorphism on $UA$,  see Remark 
\ref{rem:plusmorph}, and moreover that $g$ is a $\D$-morphism.
(4) Again we  use that the multiplication of $A$ is a bimorphism on $UA$.
(5) $+$ is a morphism of $\D$, see Remark \ref{rem:plusmorph}.
(6) Definition of the operation $\gamma$ in the product $M\times UA\times N$ 
in $\D$.
(7) Definition of the multiplication in $A^{f,g}$.\\
That $\pi_{M}$ is a $\D$-monoid morphism follows from the
computation
\begin{align*}
\pi_{M}((m,a,n)(m',a',n')) &= \pi_{M}(mm', 
f(m)\o a'+a\o g(n'),nn')\\
&= mm'\\
&= \pi_{M}(m,a,n)\pi_{M}(m',a',n').
\end{align*}
Analogously for $\pi_{N}$.
\end{proof}

\begin{proof}[Lemma \ref{lem:schuetzwelldef}]
Consider the following morphisms in $\D$:
\begin{align*}
f &\equiv (\, M \xra{\rho} M\t \one_\D \xra{M\t \iota_N} M\t N  \xra{\pi} 
{M\ast N} \xra{\eta} U  F({M\ast N})\, )\\
 g&\equiv(\, N \xra{\lambda} \one_\D\t N \xra{\iota_M\t N} M\t N  \xra{\pi} 
{M\ast N} \xra{\eta}  U F({M\ast N})\, )
\end{align*}
Note that $\pi$ is a $\D$-monoid morphism by Lemma \ref{lem:mastn}, 
$\ol\eta=\eta: M\ast N\to \ol U\,\ol F(M\ast N)$ is a $\D$-monoid morphism by 
Lemma \ref{lem:monlifting}, and that $\rho$, $\lambda$, $M\t \iota_N$ and 
$\iota_M\t N$ are $\D$-monoid morphisms follows easily from the definition of 
the monoid structure on tensor products, see Remark \ref{rem:tensormonoid}. 
Thus $f$ and $g$ and $\D$-monoid morphisms. Applying Lemma \ref{lem:afg} to 
the $\S$-algebra $A=\ol F(M\ast N)$ and the morphisms $f$ and $g$ yields the 
Sch\"utzenberger product $M\diamond N = [\ol F(M\ast N)]^{f,g}$.
\end{proof}

\begin{notation}
We denote the product projections in $\D$ by
\[
\xymatrix{
& M\diamond N \ar[dl]_{\pi_{M}} \ar[d]^{\pi_{MN}} \ar[dr]^{\pi_{N}} &\\
M & UF(M\ast N) & N
}
\]
For any $\D$-monoid morphism $f: \Psi\Sigma^* \to M\diamond N$, we put 
\[ f_M := \pi_M\c f: \Psi\Sigma^*\to M,\quad f_N:= \pi_N\c f: \Psi\Sigma^*\to N,\]
and
\[ f_{MN} := \pi_{MN}\c f: \Psi\Sigma^*\to UF(M\ast N).\]
Recall that we put $m\ast n := \pi(m\t n)$ for $m\in\under{M}$ and $n\in\under{N}$, and that $\eta: M\ast N\to UF(M\ast N)$ denotes the universal map.
\end{notation}
The following lemma appears in \cite{reutenauer79} for the case $\D=\Set$:

\begin{lemma}\label{lem:reutformula}
For any $\D$-monoid morphism $f: \Psi\Sigma^* \to M\diamond N$ and 
$u\in\Sigma^*$ we have
\[ f_{MN}(u) = \sum_{u=vaw} \eta(f_M(v)\ast 1)\o f_{MN}(a)\o \eta(1\ast 
f_N(w)),\]
where the sum ranges over all factorizations $u=vaw$ with $a\in \Sigma$ 
and $v,w\in \Sigma^*$.
\end{lemma}

\begin{proof}
The proof is by induction on the length of $u$. For $u=\epsilon$, we have 
$f(\epsilon)=(1,0,1)$ since $f$ is a $\D$-monoid morphism, and thus 
$f_{MN}(\epsilon)=0$ (the empty sum). Now suppose that the formula holds for some $u\in\Sigma^*$, and 
consider a word $ub$ with $b\in\Sigma$. Then
\begin{align*}
f_{MN}(ub) &= \pi_{MN}(f (ub)) & (1) \\
&= \pi_{MN}(f(u)f(b)) & (2)\\
&= \eta(f_M(u)\ast 1)\o f_{MN}(b) + f_{MN}(u)\o \eta(1\ast f_N(b)) & (3)\\
&= \eta(f_M(u)\ast 1)\o f_{MN}(b) &  \\
&~~~+\left[\sum_{u=vaw} \eta(f_M(v)\ast 1)\o f_{MN}(a)\o \eta(1\ast f_N(w))\right]\o  \eta(1\ast f_N(b)) & (4)\\
&= \eta(f_M(u)\ast 1)\o f_{MN}(b) &  \\
&~~~+\sum_{u=vaw} \eta(f_M(v)\ast 1)\o f_{MN}(a)\o \eta(1\ast f_N(w))\o  
\eta(1\ast f_N(b))& (5)\\
&= \eta(f_M(u)\ast 1)\o f_{MN}(b)\o \eta(1\ast f_N(\epsilon))\\
&~~~+\sum_{u=vaw} \eta(f_M(v)\ast 1)\o f_{MN}(a)\o \eta(1\ast f_N(wb)) & (6)\\
&= \sum_{ub=vaw'} \eta(f_M(v)\ast 1)\o f_{MN}(a)\o \eta(1\ast f_N(w')) & (7)\\
\end{align*}
Explanation of the individual steps:
(1) Definition of $f_{MN}$.
(2) $f$ is a $\D$-monoid morphism.
(3) Definition of the multiplication in $M\diamond N$.
(4) Induction hypothesis.
(5) Distributive law in the $\K$-algebra $\ol F(M\ast N)$.
(6) Definition of the multiplication in $M\ast N$, $\ol\eta=\eta: M\ast N \to 
\ol U\ol F(M\ast N)$ is a $\D$-monoid morphism (see \ref{lem:monlifting}, and 
$f_N$ is a $\D$-monoid morphism.
(7) Any factorization $ub=vaw'$ has either $w'=wb$ for some $w$, or $a=b$ and 
$w'=\epsilon$. 
\end{proof}

\begin{notation}
\begin{enumerate}[(1)]
\item For any two morphisms $p: M\to \S$ and $q: N\to \S$ in $\D$, we  denote 
by 
$\ol{p\ast q}: 
F(M\ast N) \to \S$ the adjoint transpose of 
$p\ast q: M\ast N \to \S$. 
\item Recall from Definition \ref{def:recog} the morphism $L_\D: 
\Psi\Sigma^*\to \S$ in $\D$ corresponding to a language $L: \Sigma^*\to S$. 
Since $L$ and $L_\D$ agree on $\Sigma^*$,  we 
usually drop the index and write $L$ for $L_\D$.
\end{enumerate}
\end{notation}

\begin{proof}[Theorem \ref{thm:schurec}]
Let $g: \Psi\Sigma^*\to M$ and $h: \Psi\Sigma^*\to N$ be $\D$-monoid 
morphisms recognizing $K$ resp. $L$. Thus there exist morphisms $p: M\to \S$ 
and $q: N\to \S$ in $\D$ with $K=p\c g$ and $L=q\c h$. Fix a letter 
$a\in\Sigma$. We define a $\D$-monoid morphism $f: \Psi\Sigma^*\to M\diamond 
N$ that recognizes the languages $K$, $L$ and $KaL$.
\begin{enumerate}[(1)]
\item Let $f: \Psi\Sigma^*\to M\diamond N$ be the unique $\D$-monoid 
morphism 
defined on letters $b\in \Sigma$ by
\[ f(b) = 
\begin{cases}
(g(b),1,h(b)), & b=a;\\
(g(b),0,h(b)), & b\neq a.
\end{cases} \] 
Then $f_M(b)=g(b)$ for all $b\in \Sigma$ and therefore $f_M = g$,
since any $\D$-monoid morphism with domain $\Psi\Sigma^*$ is determined by its 
values on the generators $\Sigma$. Similarly, we have $f_N = h$.
Since $f_{MN}(a)=1$ and $f_{MN}(b)=0$ for $b\neq a$, Lemma 
\ref{lem:reutformula} gives, for all $u\in 
\Sigma^*$,
\begin{equation}\label{eq:e12}
 f_{MN}(u) = \sum_{u=vaw}  \eta(g(v)\ast h(w)).  
\end{equation}
\item $f$ recognizes the language $K$ via the morphism $p\c \pi_M$, since
\[ K = p\c g = p\c f_M = (p\c \pi_M) \c f. \]
Analogously, $f$ recognizes $L$ via $q\c\pi_N$.
\item We show that $f$ recognizes the language $KaL$ via the morphism 
\[s \ \equiv \ (\ M\diamond N \xra{\pi_{MN}} 
UF(M\ast N)\xra{U(\ol{p\ast q})} \S\ ).\] 
Indeed, for all $u\in\Sigma^*$ we have
\begin{align*}
s\c f(u) &= \ol{p\ast q}(f_{MN}(u)) &\text{(def. $s$ and 
$f_{MN}$)} \\
&= \ol{p\ast q}\left(\sum_{u=vaw}  \eta(g(v)\ast h(w))\right) 
&\text{\eqref{eq:e12}} \\
&=\sum_{u=vaw} \ol{p\ast q}(\,\eta(g(v)\ast h(w))\,)
&\text{($\ol{p\ast q}\in \Mod{\S}$)} \\
&=\sum_{u=vaw} [p\ast q](g(v)\ast h(w))
&\text{(def. $\ol{p\ast q}$)} \\
&=\sum_{u=vaw} [\sigma\c (p\t q)](g(v)\t h(w))
&\text{(def. $p\ast q$)} \\
&=\sum_{u=vaw} \sigma(\,p(g(v))\t q(h(w))\,)
&\text{(def. $p\t q$)} \\
&=\sum_{u=vaw} \sigma(\,K(v)\t L(w)\,)
&\text{(def. $p$, $q$)} \\
&=\sum_{u=vaw} K(v)\o L(w)
&\text{(def. $\sigma$)} \\
&= (KaL)(u) & \text{(def. $KaL$)}
\end{align*}
\end{enumerate}
\end{proof}
From now on, we suppose that the Assumptions \ref{asm:dual} hold.

\begin{rem} The Assumptions \ref{asm:dual}(i) and (ii) have the following  
consequences:
\begin{enumerate}[(1)]
\item $F$ preserves finite objects. Indeed, let $D$ be a finite object of 
$\D$, and express $D$ as a 
quotient $e: \Psi X\epito D$ for some finite set $X$. Since the left adjojnt 
$F$ preserves epimorphisms, see \ref{app:adjunctions}, the morphism $Fe: 
\Pow_f X = F\Psi X \epito FD$ is an 
epimorphism in $\Mod{\S}$ and therefore surjective by Assumption 
\ref{asm:dual}(ii). Since $\Pow_f X$ is finite, so is $FD$.
\item $U$ also preserves finite objects by Remark \ref{rem:plusmorph}.
\item Consequently the Schützenberger product of two finite 
$\D$-monoids $M$ 
and $N$ is finite. To see this, recall from \ref{app:tensorproducts} that the 
tensor 
product $M\t
N$ is 
generated by the finite set $\under{M}\times \under{N}$ and is thus 
finite by Assumption \ref{asm:dual}(i). Therefore the quotient $M\ast N$ of 
$M\t N$ is also finite. 
Since $F$ and $U$ preserve finite objects by (1) and (2),
$UF(M\ast N)$ is finite.  Thus $M\diamond N$ is finite, being carried 
by the finite object $M\times 
UF(M\ast N)\times N$.
\end{enumerate}
\end{rem}

\begin{defn}\label{def:langoperations}
\begin{enumerate}[(1)]
\item The set of languages over $\Sigma$ forms an $\S$-algebra w.r.t. to the 
\emph{sum}, \emph{scalar product} and \emph{Cauchy product} of languages, 
defined by
\[ (K+L)(u) = K(u)+L(u),\, (\lambda L)(u) = \lambda L(u),\, (KL)(u) = 
\sum_{u=vw} K(v)\o L(w) \]
for languages $K, L: \Sigma^*\to S$, $\lambda \in S$, and $u\in\Sigma^*$.
Identifying a letter $a\in\Sigma$ with the language $a:\Sigma^*\to S$ that 
sends $a$ to $1$ and all other words to $0$, the marked Cauchy product 
$KaL$ is thus the Cauchy product of the languages $K$, $a$ and $L$.
\item The \emph{derivatives} of a language $L: \Sigma^*\to S$ are 
the languages $a^{-1}L,\,La^{-1}\colon \Sigma^*\to S$ ($a\in\Sigma)$ 
with
\[a^{-1}L(u) = L(au) \quad\text{and}\quad La^{-1}(u) = L(ua).\] 
\end{enumerate}
\end{defn}

\begin{lemma}\label{lem:recclosed}
Let $f: \Psi\Sigma^* \to M$ be a $\D$-monoid morphism. Then the set of 
languages recognized by $f$ is closed under sum, scalar products, and 
derivatives.
\end{lemma}

\begin{proof}
\begin{enumerate}[(1)]
\item \emph{Closure under derivatives.} Suppose that $L: \Sigma^*\to S$ is a 
language recognized by $M$, i.e. there is a morphism $p: M\to \S$ in $\D$ with 
$L=p\c 
f$. We claim that $f$ recognizes the left derivative $a^{-1}L$ via the 
morphism $p' := p\c (f(a)\bullet \dash)$. (Note that $f(a)\bullet\dash: M\to 
M$ is a morphism of $\D$ because $\bullet$ is a $\D$-bimorphism.) Indeed, we 
have for all $u\in\Sigma^*$:
\[ p'(f(u)) = p(f(a)\bullet f(u)) = p(f(au)) = L(au) = (a^{-1}L)(u). \]
and thus $L=p'\c f$, as claimed. Analogously for right derivatives.
\item \emph{Closure under sums.} 
Let $K$ and $L$ be two languages recognized by $f$, i.e. $K = p\c f$ and 
$L=q\c f$ for morphisms $p, q: M\to \S$ in $\D$. Denote by $\ol p, \ol q: FM\to \S$ 
 the adjoint transposes of $p$ and $q$ in $\Mod{\S}$ (w.r.t. the adjunction 
$F\dashv U$). Since $\Mod{\S}$ forms a commutative variety, we have the morphism $\ol p + 
\ol q: FM\to \S$ in $\Mod{\S}$, defined by $[\ol p+\ol q](x) = \ol p(x) + \ol 
q(x)$. Then 
$K+L$ is recognized by the morphism $U(\ol{p}+\ol{q})\c \eta: M\to \S$, since 
for all $u\in\Sigma^*$,
\[ [\ol{p}+\ol{q}](\eta(f(u))) = \ol p(\eta(f(u))) +\ol q(\eta(f(u))) = 
p(f(u))+q(f(u)) = 
K(u)+L(u). \]
Thus $K+L = U(\ol p+\ol q)\c \eta\c f$, i.e. $f$ recognizes $K+L$.
\item \emph{Closure under scalar product.} Analogous to (2).
\end{enumerate}
\end{proof}

\begin{proof}[Theorem \ref{thm:universalprop}]
We first establish two preliminary 
technical results (steps (1) and (2)). 
\begin{enumerate}[(1)]
\item Consider the commutative diagram below, where $\ol t$ is the adjoint 
transpose of the universal bimorphism $t$; note that $S^{(\dash)} = F\Psi$ 
since 
the adjunction $F\dashv U: \Mod{\S}\to \D$ is concrete.
\[
\xymatrix{
\under{\S^{(\under{M}\times\under{N})}} \ar@{->>}[r]^{F\ol t} & \under{F(M\t 
N)} 
\ar@{->>}[r]^{F\pi} & \under{F(M\ast N)}\\
\under{\Psi(\under{M}\times\under{N})} \ar[u]^\eta \ar@{->>}[r]^{\ol t} & 
\under{M\t N} \ar[u]_\eta  \ar@{->>}[r]^\pi & \under{M\ast N} \ar[u]_\eta \\
\under{M}\times \under{N} \ar@{>->}[u] \ar[ur]_t &&
}
\]
Since $F\dashv U$ is a concrete monoidal adjunction, precomposing the unit 
$\eta: 
\under{\Psi(\under{M}\times\under{N})} \to 
\under{\S^{(\under{M}\times\under{N})}}$ 
with the unit $\under{M}\times\under{N}\monoto 
\under{\Psi(\under{M}\times\under{N})}$ of the adjunction $\Psi\dashv 
\under{\dash}_\D: \D\to \Set$ gives the unit
$\under{M}\times\under{N}\monoto \under{\S^{(\under{M}\times\under{N})}}$ of  
the adjunction
$\S^{(\dash)} \dashv\under{\dash}_\S: \Mod{\S}\to \Set$. 
Thus the $\S$-linear map $F\pi\c F\ol t$ maps a generator $(m,n)$ of 
$\S^{(\under{M}\times\under{N})}$ to $\eta(m\ast n)$. Since the left adjoint 
$F$ preserves epimorphisms, $F\pi$ and  $F\ol t$ are surjective by Assumption 
\ref{asm:dual}(ii). It 
follows that every element of the $\S$-module $F(M\ast N)$ can be expressed as 
a linear 
combination $\sum_{j=1}^n \lambda_j \eta(m_j\ast n_j)$ with $\lambda_j\in 
S$, $m_j\in \under{M}$ and $n_j\in\under{N}$.
\item Let $p: M\to \S$ and $q: N\to \S$ be two morphisms in $\D$, and let $L$ 
be 
the language recognized by $f$ via $U(\ol{p\ast q})\c \pi_{MN}$, i.e. 
\begin{equation}\label{eq:lpq} L := 
U(\ol{p\ast q})\c \pi_{MN}\c f = U(\ol{p\ast q})\c f_{MN}.\end{equation}
By (1), each element $f_{MN}(a)\in\under{F(M\ast N)}$ with $a\in\Sigma$ can be 
expressed as a linear combination
\begin{equation}\label{eq:lincomb} f_{MN}(a) = \sum_{j=1}^{n_a} 
\lambda_j^a\eta(m^{a}_j\ast n^{a}_j).  \end{equation}
with $\lambda_j^a\in S$, $m^{a}_j\in \under{M}$ and $n^a_j\in\under{N}$. For 
$a\in\Sigma$ and $j=1,\ldots, n_a$, put
\begin{equation}\label{eq:l1l2} L^{a,j}_M := p \c (\dash\bullet m^{a}_j) \c 
f_M \quad\text{and}\quad   L^{a,j}_N := q \c (n^{a}_j\bullet \dash) \c 
f_N,\end{equation}
where $\dash\bullet m^{a}_j: M\to M$ and $n^{a}_j\bullet \dash: N\to N$ are 
morphisms in $\D$ 
because the monoid multiplication of $M$ resp. $N$ is a $\D$-bimorphism. Then 
$L$ can be expressed as the following linear combination of 
languages (cf. Definition \ref{def:langoperations}):
\begin{equation}\label{eq:recformula}
L = \sum_{a\in\Sigma}\sum_{j=1}^{n_a} \lambda_j^a(L_M^{a,j}aL_N^{a,j}).
\end{equation}
To prove this, we compute for all $u\in\Sigma^*$:
\begin{align*}
L(u) &= \ol{p\ast q}(f_{MN}(u)) & (1)\\
&= \ol{p\ast q}\left(\sum_a \sum_{u=vaw} \eta(f_M(v)\ast 1)\o f_{MN}(a)\o 
\eta(1\ast f_N(w))\right) & (2)\\
&= \ol{p\ast q}\left(\sum_a \sum_{u=vaw} \eta(f_M(v)\ast 
1)\o\left(\sum_{j=1}^{n_a} 
\lambda_j^a\eta(m^{a}_j\ast n^{a}_j)\right)\o 
\eta(1\ast f_N(w))\right) & (3)\\
&= \ol{p\ast q}\left(\sum_a \sum_{u=vaw}
\sum_{j=1}^{n_a} \lambda_j^a \eta(f_M(v)\ast 
1)\o 
\eta(m^{a}_j\ast n^{a}_j)\o 
\eta(1\ast f_N(w))\right) & (4)\\
&= \ol{p\ast q}\left(\sum_a \sum_{u=vaw}
\sum_{j=1}^{n_a} \lambda_j^a \eta(\,(f_M(v)m^a_j)\ast (n^a_j 
f_N(w)\,)\,)\right) & (5)\\
&= \ol{p\ast q}\left(\sum_a \sum_{j=1}^{n_a} \lambda_j^a \left(\sum_{u=vaw} 
\eta(\,(f_M(v)m^a_j)\ast (n^a_j f_N(w)\,)\,)\right)\right) & (6)\\
&= \sum_a \sum_{j=1}^{n_a} \lambda_j^a \left(\sum_{u=vaw} 
\ol{p\ast q}(\,\eta(\,(f_M(v)m^a_j)\ast (n^a_j f_N(w)\,)\,)\,)\right) & (7)\\
&= \sum_a \sum_{j=1}^{n_a} \lambda_j^a \left(\sum_{u=vaw} 
[p\ast q](\,(f_M(v)m^a_j)\ast (n^a_j f_N(w)\,)\,)\right) & (8)\\
&= \sum_a \sum_{j=1}^{n_a} \lambda_j^a \left(\sum_{u=vaw} 
[\sigma\c (p\t q)](\,(f_M(v)m^a_j)\t (n^a_j f_N(w)\,)\,)\right) & (9)\\
&= \sum_a \sum_{j=1}^{n_a} \lambda_j^a \left(\sum_{u=vaw} 
p(f_M(v)m^a_j))\o q(n^a_j f_N(w))\right) & (10)\\
&=\sum_a \sum_{j=1}^{n_a}  \lambda_j^a \left(\sum_{u=vaw} L_M^{a,j}(v) 
L_N^{a,j}(w)\right) & (11) \\
&=\sum_a \sum_{j=1}^{n_a}  \lambda_j^a (L_M^{a,j}aL_N^{a,j})(u)  & (12) 
\end{align*}
Explanation of the individual steps: 
(1) Definition of $L$.
(2) Lemma \ref{lem:reutformula}.
(3) Apply equation \eqref{eq:lincomb}.
(4) The multiplication in the $\S$-algebra $\ol F(M\ast N)$ distributes over 
sum and 
scalar product.
(5) Use the definition of the multiplication in $M\ast N$, and the fact that 
$\eta=\ol\eta: \ol U \ol F(M\ast N)$ is a $\D$-monoid morphism by Lemma 
\ref{lem:monlifting}.
(6) Interchange sums, and use that the scalar product distributes over sums.
(7) $\ol{p\ast q}$ is an $\S$-linear map. 
(8) Definition of $\ol{p\ast q}$.
(9) Definition of $p\ast q$.
(10) Definition of $\sigma$ and $p\t q$.
(11) Definition of $L_M^{a,j}$ and $L_N^{a,j}$.
(12) Definition of marked Cauchy product and scalar product of languages.
\item We are prepared to prove the theorem. Suppose that $e: 
\Psi\Sigma^*\epito P$ is a surjective $\D$-monoid morphism satisfying the 
assumptions of the theorem. For any $p: M\to\S$ in $\D$, the language 
$(\Sigma^*\monoto \Psi\Sigma^* \xra{f} M\ast N \xra{\pi_{M}} M \xra{p}\S)$ 
is recognized by $\pi_M\c f$  and hence, by 
assumption, also by $e$. Consequently there exists a morphism $h^p: P\to \S$ 
in $\D$  with $h^p\c e = p\c \pi_M\c f$. Analogously, for each $q: N\to \S$ 
there exists a morphism $h^q: P\to \S$ in $\D$ with $h^q\c e = q\c \pi_N\c f$.
Moreover, for any pair of morphisms $p: M\to \S$ and $q: N\to \S$ in $\D$, the 
language $L$ of \eqref{eq:lpq} is recognized by $e$. Indeed, by definition the 
languages $L^{a,j}_M$ and $L^{a,j}_M$ of \eqref{eq:l1l2} are recognized by 
$f_M=\pi_M\c f$ 
resp. $f_N = \pi_N\c f$. Thus, by the assumptions on $e$, the marked 
product 
$L^{a,j}_M a L^{a,j}_M$ is recognized by $e$ for every $a\in \Sigma$ and 
$j=1,\ldots, n_a$. 
Therefore, by Lemma 
\ref{lem:recclosed},
$e$ recognizes $L$. That is, there exists a morphism $h^{p,q}: P\to \S$ in 
$\D$ with $h^{p,q}\c e = L$ ($= U(\ol{p\ast q})\c \pi_{MN}\c f$).

To summarize, we have the following commutative diagrams for all $p: M\to \S$ 
and $q: N\to \S$:
\[
\xymatrix{
\Psi\Sigma^* \ar[d]_f  \ar@{->>}[rr]^e && P \ar[d]^{h^p}\\
M\diamond N \ar[r]_{\pi_M} & M \ar[r]_{p} & \S
}\quad 
\xymatrix{
\Psi\Sigma^* \ar[d]_f  \ar@{->>}[rr]^e && P \ar[d]^{h^q}\\
M\diamond N \ar[r]_{\pi_N} & N \ar[r]_{q} & \S
}
\]
\[
\xymatrix{
\Psi\Sigma^* \ar[d]_f  \ar@{->>}[rr]^e && P \ar[d]^{h^{p,q}}\\
M\diamond N \ar[r]_<<<<{\pi_{MN}} & UF(M\ast N) \ar[r]_<<<<<{U(\ol{p\ast q})} 
& \S
}
\]
Since the family $\{\pi_M,\pi_N,\pi_{MN}\}$ of product projections is 
separating, and using Assumption 
\ref{asm:dual}(iii), the combined family
\begin{align*} 
& \{\,M\diamond N\xra{\pi_M} M\xra{p} \S\,\}_{p: M\to \S}\\
&\cup \{\,M\diamond 
N\xra{\pi_N} N\xra{q} \S\,\}_{q: N\to \S}\\
& \cup \{\,M\diamond 
N\xra{\pi_{MN}} UF(M\ast N)\xra{U(\ol{p\ast q})} \S\,\}_{p: M\to \S,\,q: N\to \S}
\end{align*}
is  separating, see \ref{app:jointlyinj}. Thus diagonal fill-in, see 
\ref{app:factsystems}, gives a unique $h: P\to M\diamond N$ in $\D$ with $h\c 
e = f$. Moreover, since both $e$ and $f$ are $\D$-monoid morphisms and $e$ is 
surjective, $h$ is a $\D$-monoid morphism.
\end{enumerate}
\end{proof}

\begin{defn}\label{rem:recclosed}
A language is called \emph{recognizable} if it is recognized by some 
finite $\D$-monoid. Denote by $\Rec{\Sigma}$ the set of $\D$-recognizable 
languages over the alphabet $\Sigma$. 
\end{defn}

\begin{rem}
Recognizable languages are precisely the regular languages $L: \Sigma^*\to S$, 
i.e. languages accepted by some finite Moore automaton with output set $S$; 
see \cite{uacm16}[Lemma G.1].
\end{rem}

In \cite{ammu14} we established (working with a subset of our present 
Assumptions \ref{asm:dual}) the following results:

\begin{theorem}[Ad\'amek, Milius, Myers, Urbat 
\cite{ammu14}]\label{thm:rec_properties}
 $\Rec{\Sigma}$ forms an algebra in the variety $\C$ w.r.t to the 
 $\C$-algebraic operations on languages. Moreover, $\Rec{\Sigma}$ is closed 
 under derivatives, and the 
maps $a^{-1}(\dash)$ and $(\dash)a^{-1}$ on $\Rec{\Sigma}$ are morphisms 
of $\C$, i.e. the $\C$-algebraic operations preserve derivatives.
\end{theorem}

\begin{theorem}[Ad\'amek, Milius, Myers, Urbat 
\cite{ammu14}]\label{thm:locvar_properties}
For any finite set $\V_\Sigma$ of recognizable languages over $\Sigma$, the 
following statements are equivalent:
\begin{enumerate}[(i)]
\item $\V_\Sigma$ is closed under the $\C$-algebraic operations and 
derivatives. 
\item There exists a finite $\D$-monoid $P$ and a surjective $\D$-monoid morphism $e: \Psi\Sigma^*\epito P$ 
 such that $e$ recognizes precisely the 
languages in $\V_\Sigma$.
\end{enumerate}
\end{theorem}

\begin{proof}[Theorem \ref{thm:schurec2}]
Let $\V_\Sigma$ be the closure of the set $\L_{M,N}(f)$ under $\C$-algebraic 
operations and derivatives. We first show that $\V_\Sigma$ is finite. Indeed, 
since $\L_{M,N}(f)$ is a finite set, and every recognizable language has 
only finitely many derivatives by Lemma \ref{lem:recclosed}, the closure 
$\V_\Sigma'$ of $\L_{M,N}(f)$ under derivatives is a finite set. By Theorem 
\ref{thm:rec_properties}, the closure of $\V_\Sigma'$ under $\C$-algebraic 
operations is again closed under derivatives, and thus equal to $\V_\Sigma$. 
Since 
$\C$ is a locally finite variety by Assumption \ref{asm:dual}(iv), this shows 
that $\V_\Sigma$ is finite. Therefore, by 
Theorem 
\ref{thm:locvar_properties}, there exists a finite $\D$-monoid $P$ and a 
surjective $\D$-monoid morphism 
$e: \Psi\Sigma^*\epito P$ such that $e$ 
recognizes 
precisely the languages of $\V_\Sigma\supseteq \L_{M,N}(f)$. Theorem 
\ref{thm:universalprop} 
gives a $\D$-monoid morphism $h: P\to M\diamond N$ with $h\c e = f$. 
Then every language recognized by $f$ (say via the morphism $s: M\diamond 
N\to 
\S$ in $\D$) 
is also recognized by $e$ (via the morphism $s\c h$), and therefore lies in 
$\V_\Sigma$. 
\end{proof}

\begin{rem}
For the categories $\C/\D$ of Example \ref{ex:predualcats}) one can drop 
the 
closure under derivatives in Theorem \ref{thm:schurec2}: in each 
case, the closure $\V_\Sigma$ of $\L_{M,N}(f)$ under $\C$-algebraic 
operations is already closed under derivatives. To see this, note that by 
Theorem 
\ref{thm:rec_properties} it suffices to show that every derivative of a 
language in 
$\L_{M,N}(f)$ lies in $\V_\Sigma$. The latter is clear for the languages $K$ and 
$L$ recognized by $\pi_M\c f$ resp. $\pi_N\c f$: by Lemma \ref{lem:recclosed} 
their derivatives are even elements of $\L_{M,N}(f)$. 
Now consider the languages of the form $KaL$ in $\L_{M,N}(f)$. One easily 
verifies that
\[
b^{-1}(KaL) =
\begin{cases}
(b^{-1}K)aL,& b\neq a;\\
(a^{-1}K)aL + K(\epsilon)L, & b=a,
\end{cases}
\]
and analogously for right derivatives. In the case $\D=\Vect{\K}$, this 
shows that any derivative of $KaL$ is a linear combination (i.e. a 
$\Vect{\K}$-algebraic combination) of languages in 
$\L_{M,N}(f)$, and thus lies in $\V_\Sigma$. For the other examples 
($\D=\Set,\Pos,\JSL$ with $\S=\{0,1\}$), the above case $b=a$ 
states that 
\[a^{-1}(KaL) = \begin{cases}
(a^{-1}K)aL & \epsilon\not\in L;\\
(a^{-1}K)aL \cup L & \epsilon\in L.
\end{cases}
\]
Thus every derivative of $KaL$ is a finite union of languages in $\L_{M,N}(f)$ 
and 
therefore lies in $\V_\Sigma$, since the union is part of the $\C$-algebraic 
operations 
for $\D=\Set,\Pos,\JSL$. 
\end{rem}

\section{Details for the Examples}
\textbf{Details for Example \ref{ex:monoidalfunc}.}
\begin{enumerate}[(1)]
\item See Example \ref{ex:upsimonoidal}.
\item This is a special case of Example \ref{ex:upsimonoidal} with 
$\D=\JSL$ 
and $\Psi = \Pow_f$ (cf. Example \ref{ex:psi}). Explicitly, $\Pow_f$ is 
monoidal 
w.r.t. the isomorphism $\theta_{X,Y}: \Pow_f X\t \Pow_f Y \cong \Pow_f(X\times 
Y)$ 
whose inverse $\theta_{X,Y}^{-1}$ maps $\{(x_1,y_1),\ldots, 
(x_n,y_n)\}\seq X\times
Y$ to $(\vee_{j=1}^n x_j\t y_j) \in X\t Y$, and the 
identity morphism $\theta_1 = \id: \one_\JSL \to \Pow_f 1$.
\item Let us first verify that $\Dn_f$ is a left adjoint to $U$.
 The unit is given by the monotone map $\eta_X: X\to U\Dn_f(X)$ with $\eta(x)=\dnarrow \{x\}$. 
Given a monotone map $h: X\to UA$ into a semilattice $A$, let $\ol 
h: \Dn_f(X)\to A$ be the function that maps a
finitely generated down-set $S=\dnarrow S_0$ of $X$ to the finite join 
$\bigvee S = \bigvee S_0\in \under{A}$. Clearly $\ol h$ is a semilattice morphism and satisfied $U(\ol h)\c \eta_X = h$. Moreover, since every finitely generated down-set is a finite union of one-generated downsets $\dnarrow \{x\}$, $\ol h$ is uniquely determined by this property.

To show that $U$ is monoidal, observe that for any two semilattices $A$ and 
$B$, the universal bimorphism $t_{A,B}: \under{A}\times\under{B}\to \under{A\t 
B}$ is monotone, since it preserves joins in each component (and is thus 
monotone in each component). Thus we can view $t_{A,B}$ as a morphism 
$t_{A,B}^*: UA\times UB \to U(A\t B)$ in $\Pos$, and in complete analogy to 
Example \ref{ex:upsimonoidal}(1) one can show that $U$ is monoidal w.r.t. the 
maps $t_{A,B}^*$ and the unit $\eta: 1\mapsto UF1 = U\one_\JSL$, where 
$1=\one_\Pos$ is the one-element poset.  Similarly, 
the proof that $\Dn_f$ is monoidal is analogous to Example 
\ref{ex:upsimonoidal}(2), with $\Set$ replaced by $\Pos$ and $\under{\dash}$, 
$\Psi$ by $U$, $\Dn_f$.  Explicitly, $\Dn_f$ is 
monoidal 
w.r.t. the isomorphism $\theta^*_{X,Y}: \Dn_f X \t \Dn_f Y \cong \Dn_f(X\times 
Y)$ 
whose inverse $(\theta_{X,Y}^*)^{-1}$ maps a down-set $\{(x_1,y_1),\ldots, 
(x_n,y_n)\}\seq X\times
Y$ to $(\vee_{j=1}^n x_j\t y_j) \in X\t Y$, and the 
identity morphism $\theta_1 = \id: \one_\JSL \to \Dn_f 1$.
\end{enumerate}
~\\
\noindent\textbf{Details for Example \ref{ex:monoidaladj}.}
The adjunction $\Psi\dashv \under{\dash}: \D\to\Set$ is monoidal by Example 
\ref{ex:upsimonoidal}. The proof given in that example also works for $U\dashv 
\Dn_f: \JSL\to \Pos$, replacing $\D$ by $\JSL$ and $\Set$ by $\Pos$. That 
$\Id\dashv\Id: \D\to\D$ is a monoidal adjunction is trivial. \\

\noindent\textbf{Details for Example \ref{ex:monoidaladj}.}
(3), (4), (5) are trivially concrete monoidal adjunctions. So is $\Pow_f\dashv 
\under{\dash}: \JSL\to \Set$, since the monoidal adjunction of $\Set$ is the 
identity adjunction $\Id\dashv \Id: \Set\to\Set$. Concerning $\Dn_f\dashv U: 
\JSL\to \Pos$, we clearly have $\under{\dash}_\JSL = \under{\dash}_\Pos \c U$ 
and $\Pow_f = \Dn_f \c \Psi_\Pos$. Since also the units and counits of these 
three adjunctions compose accordingly, the composite of the adjunction $\Dn_f 
\dashv U$ with 
$\Psi_\Pos \dashv \under{\dash}_\D$ is the adjunction $\Pow_f \dashv 
\under{\dash}_\JSL$.\\

\noindent\textbf{Details for Example \ref{ex:jointlym}.}
\begin{enumerate}[(1)]
\item For $\D=\Set$ or $\Pos$ and $\S=\{0,1\}$, the family $\{\,M\times N 
\xra{p\times q} \{0,1\}\times \{0,1\} \xra{\sigma} \{0,1\}\,\}_{p,q}$ is 
separating, where $\sigma(m,n)=m\o n$ is the multiplication of $S$. We prove this only for $\D=\Pos$, the argument 
for $\D=\Set$ being analogous. Let $(m,n)$ and $(m',n')$ be two 
elements of $M\times N$ with $(m,n)\not\leq(m',n')$, say 
$m\not\leq m'$. Choose $p: M\to \{0,1\}$ to be the monotone map with 
$p(x)=0$ iff $x\leq m'$, and $q: N\to\{0,1\}$ to be the constant map 
on $1$. Then
\[\sigma\c (p\times q)(m,n) = p(m)\o q(n) = 1\o 1 = 1\] and \[\sigma\c 
(p\times q)(m',n') = p(m')\o q(n') = 0\o 1 = 0,\] so $\sigma\c 
(p\times q)(m,n) \not\leq \sigma\c(p\times q)(m',n')$. This shows 
that
the family $(\sigma\c (p\times q))_{p,q}$ is separating.
\item Similarly, for $\D=\Vect{\K}$ the family $\{\,M\t N 
\xra{p\times q} \K\t \K \xra{\sigma} \K\,\}_{p,q}$ is separating. To see this, 
recall that if $M$ and $N$ are a vector spaces with bases $\{b_i\}_{i\in I}$ 
resp. $\{c_j\}_{j\in J}$, then 
the 
tensor product $M\t 
N$ has the basis $\{b_i\t c_j\}_{i\in I, j\in J}$. It suffices to 
show that, for any element $x\in 
\under{M\t 
N}$ with $x\neq 0$, there are linear maps $p: M\to \K$ and $q: N\to \K$ with 
$\sigma\c 
(p\t q)(x)\neq 0$. Suppose that $x=\sum_{i\in I, j\in J} 
\lambda_{i,j} (b_{i}\t c_j)$ with $\lambda_{i,j}\in \K$. 
Since $x\neq 0$, there exist $i_0\in I$ and $j_0\in J$ with 
$\lambda_{i_0,j_0}\neq 0$. Let $p: M\to \K$ and $q:N\to\K$ be the 
linear maps defined by
\[ p(b_i) =
\begin{cases}
1,&i=i_0;\\
0,&i\neq i_0,
\end{cases}
\quad\text{and}\quad
p(c_j) =
\begin{cases}
1,&j=j_0;\\
0,&j\neq j_0.
\end{cases}
 \]
Then
\begin{align*}
\sigma\c (p\t q)(x) = \sum_{i,j} \lambda_{i,j} [\sigma\c 
 (p\t q)(b_i\t c_j)]
 = \sum_{i,j} \lambda_{i,j} 
 (p(b_i)\o q(c_j))
 = \lambda_{i_0,j_0}
 \neq 0.
\end{align*}
Here we use the linearity of $\sigma$ and $p\t q$ in the first step, the definition of $\sigma$ and $p\t q$ in the second step, and the definition of $p$ and $q$ in the last step.
\item  If $\D=\JSL$ and $M$ and $N$ are finite idempotent semilattices, we can 
describe $M\ast N$ as 
follows. Consider the surjective semilattice morphism
\[ e\,\equiv\,(\, \Pow_f(\under{M}\times\under{N}) \xra{\ol t} M\t N 
\xra{\pi} M\ast N \,), \]
where $\ol t$ is the adjoint transpose of the universal bimorphism $t: 
\under{M}\times \under{N}\to \under{M\t N}$; cf. \ref{app:tensorproducts}. Let $\equiv\seq  
\Pow_f(\under{M}\times\under{N})\times  
\Pow_f(\under{M}\times\under{N})$ be the kernel of $e$, i.e. $X\equiv Y$ 
iff $\pi\c \ol t (X) = \pi\c \ol t (Y)$. By the definition of $\pi$, the latter means precisely that $\sigma\c 
(p\t q)\c \ol t(X) = 
\sigma\c
(p\t q)\c \ol t(Y)$ for all semilattice morphisms $p: M\to \{0,1\}$ 
and $q: 
N\to \{0,1\}$. Since such morphisms correspond to ideals, see Example 
\ref{ex:monrec}(3), we have $X\equiv Y$ 
iff, for all ideals $I\seq M$ and $J\seq N$,
\[ \exists (m,n)\in X: m\not\in I \wedge n\not\in 
J\quad\Lra\quad \exists (m',n')\in Y: m'\not\in I \wedge 
n'\not\in J.\]
Since $\equiv$ is a semilattice congruence on 
$\Pow_f(\under{M}\times\under{N})$, being the kernel of a semilattice 
morphism, $X\equiv Y$ and $X\equiv Z$ implies $X=X\cup X \equiv Y\cup Z$. 
Thus, for every $X\seq \under{M}\times\under{N}$ there exists a largest 
set $[X]\seq \under{M}\times\under{N}$ with $X\equiv [X]$, viz. 
the union of all $Y$ with $X\equiv Y$. It follows that $X\mapsto [X]$ defines 
a 
closure operator on $\Pow_f(\under{M}\times\under{N})$, and that every 
equivalence class of $\equiv$ contains a unique closed subset of 
$\under{M}\times\under{N}$ (viz. the union of all sets in the equivalence 
class). One easily 
verifies that $[X]$ consists of those elements $(x,y)\in 
\under{M}\times\under{N}$ such that, for all ideals $I\seq M$ and $J\seq N$,
\[ x\not\in I\wedge y\not\in J\quad\Ra\quad \exists(m,n)\in X: m\not\in I\wedge n\not\in J. \]
Since $\equiv$ is the kernel of $\pi\c \ol t$, we have
 $M\ast N 
\cong\Pow_f(\under{M}\times\under{N})/\mathord{\equiv}$. Identifying the 
equivalence classes of $\equiv$ with the closed subsets of $M\times N$, the 
join in the idempotent semiring $M\ast N$ is given by $[X]\vee [Y]=[X\cup Y]$, and the multiplication 
(see Lemma \ref{lem:mastn}) by $[X][Y]=[XY]$. 
\end{enumerate}
~\\
\noindent\textbf{Details for Example \ref{ex:predualcats}}
Clearly Assumption \ref{asm:dual}(i) holds for our examples
$\D=\Set$, $\Pos$, $\JSL$ and $\Vect{\K}$ ($\K$ a finite field). Also (ii) is 
well-known in all these cases. For $\Set$, $\Pos$ and $\Vect{\K}$, see e.g. 
\cite[Example 7.40]{adherstr}. For $\JSL$, see \cite{hornkimuara}.  Concerning 
(iii), that $\D(M,\S)$ 
and $\D(N,\S)$ are separating is easy to 
verify in all cases. Also, that $\{\,U(\ol{p\ast q}) \}_{p,q}$ forms a 
separating family is trivial for $\D=\JSL$ and 
$\Vect{\K}$, since here $U=\Id$ and $U(\ol{p\ast q}) = p\ast q$, and the 
morphisms $p\ast q$ are separating by definition. 

It remains to  consider the cases $\D=\Set$ and $\D=\Pos$. We only treat 
$\Pos$, the argument for $\Set$ being the discrete special case. We 
need to show 
that the family of monotone maps
\[\{\,\ol{p\times q}: \Dn_f (M\times N)\to \{0,1\}\,\}_{p: 
M\to\{0,1\},\,q:N\to\{0,1\}}\] is separating, where $\S=\{0,1\}$ 
(considered as a poset) is ordered by $0<1$.  
Note that $\ol{p\times q}$ maps a finitely generated down-set $X\seq M\times 
N$ to $1$ iff 
there exists a pair $(m,n)\in X$ with $p(m)=1$ and $q(n)=1$.
Let $X\not\seq Y$ be two finitely generated down-sets of $M\times N$. 
Choose an element $(m,n)\in 
X\setminus 
Y$, and define monotone maps $p: M\to \{0,1\}$ and $q: N\to \{0,1\}$ by
\[ \text{$p(x)=1$ iff $x\geq m$}\quad\text{resp.}\quad  \text{$q(y)=1$ iff 
$y\geq n$}.\]
Then we get
\[ \ol{p\times q}(X) = 1\not\leq 0 = \ol{p\times q}(Y), \]
i.e. $\ol{p\times q}$ separates $X$ and $Y$, as desired.
